%% file: JournalV20.tex
\documentclass[12pt,draftcls,onecolumn]{IEEEtran}
\usepackage{amsmath,graphicx,float,cite,dsfont,amssymb}
\usepackage{amsthm,amsfonts,multirow,bm,array,setspace,stfloats}
\usepackage{amsmath}
\usepackage{mathtools}
\newcommand{\defeq}{\vcentcolon=}
\newcommand{\eqdef}{=\vcentcolon}
\usepackage{textcomp}
\usepackage{subfigure}
\usepackage{psfrag}
\usepackage{color}
\usepackage{pgfplots}
\usepackage{graphicx}
\usepackage{listings}
\usepackage{tikz}
\usetikzlibrary{arrows}
\usetikzlibrary{patterns}
\tikzset{>=latex'}
\usetikzlibrary{shapes.geometric}
\usepackage{multirow}
\usepackage{graphicx,subfigure}
\usetikzlibrary{calc}
\usetikzlibrary{positioning, shapes}
\usepackage{algorithm}
\usepackage{algpseudocode}
\usepackage{mathtools}
\usepackage[skip=2pt,font=small]{caption}
\usepackage{tabularx,ragged2e,booktabs,caption}
\usepackage{balance}
\makeatletter
\def\endthebibliography{%
  \def\@noitemerr{\@latex@warning{Empty `thebibliography' environment}}%
  \endlist
}
\usepackage[utf8]{inputenc}
\usepackage[english]{babel}

\usepackage{amsthm}
\DeclarePairedDelimiter\normV{\lVert}{\rVert}
\newcommand{\norm}[1]{\left\lVert#1\right\rVert}
\theoremstyle{remark}
\newtheorem{theorem}{ {Theorem}}

\newtheorem{proposition}{{Proposition}}

\DeclareMathOperator{\EX}{\mathbb{E}}% expected value
						
\def\hexagonsize{1cm}
\pgfdeclarepatternformonly
{hexagons}% name
{\pgfpointorigin}% lower left
{\pgfpoint{3*\hexagonsize}{0.866025*2*\hexagonsize}}%  upper right
{\pgfpoint{3*\hexagonsize}{0.866025*2*\hexagonsize}}%  tile size
{% shape description
	\pgfsetlinewidth{0.4pt}
	\pgftransformshift{\pgfpoint{0mm}{0.866025*\hexagonsize}}
	\pgfpathmoveto{\pgfpoint{0mm}{0mm}}
	\pgfpathlineto{\pgfpoint{0.5*\hexagonsize}{0mm}}
	\pgfpathlineto{\pgfpoint{\hexagonsize}{-0.866025*\hexagonsize}}
	\pgfpathlineto{\pgfpoint{2*\hexagonsize}{-0.866025*\hexagonsize}}
	\pgfpathlineto{\pgfpoint{2.5*\hexagonsize}{0mm}}
	\pgfpathlineto{\pgfpoint{3*\hexagonsize+0.2mm}{0mm}}
	\pgfpathmoveto{\pgfpoint{0.5*\hexagonsize}{0mm}}
	\pgfpathlineto{\pgfpoint{\hexagonsize}{0.866025*\hexagonsize}}
	\pgfpathlineto{\pgfpoint{2*\hexagonsize}{0.866025*\hexagonsize}}
	\pgfpathlineto{\pgfpoint{2.5*\hexagonsize}{0mm}}
	\pgfusepath{stroke}
}
\newcommand{\iPhone}[3]{
	\coordinate (a) at (#1,#2);
	\draw [line width=0.25pt,rounded corners=(#3)*1mm,fill=white,scale=(#3)] (a)--($(a)+(0.67,0)$)--($(a)+(0.67,1.381)$)--($(a)+(0,1.381)$)--cycle;
	\draw [color=gray,line width=0.25pt,rounded corners=(#3)*0.8mm,fill=white,scale=(#3)] ($(a)+(0.015,0.015)$)--($(a)+(0.655,0.015)$)--($(a)+(0.655,1.366)$)--($(a)+(0.015,1.366)$)--cycle;
	\draw [line width=0.25pt,rounded corners=(#3)*0.04mm,scale=(#3)] ($(a)+(0.2875,1.266)$)--($(a)+(0.3825,1.266)$)--($(a)+(0.3825,1.281)$)--($(a)+(0.2875,1.281)$)--cycle;
	\draw[line width=0.25pt,scale=#3] ($(a)+(0.335,0.09)$) circle (0.055cm);
	\draw[line width=0.25pt,scale=#3] ($(a)+(0.335,0.09)$) circle (0.044cm);
	\draw[line width=0.25pt,scale=#3] ($(a)+(0.2275,1.2735)$) circle (0.015cm);
	\draw[line width=0.25pt,scale=#3] ($(a)+(0.335,1.32)$) circle (0.01cm);
	\draw [fill={rgb:black,1;white,4},line width=0.25pt,scale=(#3)] ($(a)+(0.042475,0.170195)$)--($(a)+(0.042475,0.170195)+(0.58505,0.0)$)--($(a)+(0.042475,0.170195)+(0.58505,1.04061)$)--($(a)+(0.042475,0.170195)+(0.0,1.04061)$)--cycle;
}
\newcommand{\basestation}[3]{
	\coordinate (a) at (#1,#2);
	\draw[line width=(#3)*1.5pt,scale=#3] ($(a)+(0, -1.08)$) -- ($(a)+(0, 1)$);
	\draw[line width=(#3)*1.5pt,scale=#3] ($(a)+(0,1)$) .. controls ($(a)+(-0.25,-0.5)$) .. ($(a)+(-0.5,-0.9)$);
	\draw[line width=(#3)*1.5pt,scale=#3] ($(a)+(0,1)$) .. controls ($(a)+(0.25,-0.5)$) .. ($(a)+(0.5,-0.9)$);
	\draw[line width=(#3)*1.5pt,scale=#3] ($(a)+(0, -0.9)$) -- ($(a)+(-0.35, -0.7)$);
	\draw[line width=(#3)*1.5pt,scale=#3] ($(a)+(0, -0.9)$) -- ($(a)+(0.35, -0.7)$);
	\draw[line width=(#3)*0.75pt,scale=#3] ($(a)+(-0.35, -0.65)$) -- ($(a)+(0, -0.5)$);
	\draw[line width=(#3)*0.75pt,scale=#3] ($(a)+(0.35, -0.65)$) -- ($(a)+(0, -0.5)$);
	\draw[line width=(#3)*1pt,scale=#3] ($(a)+(0, -0.6)$) -- ($(a)+(-0.3, -0.45)$);
	\draw[line width=(#3)*1pt,scale=#3] ($(a)+(0, -0.6)$) -- ($(a)+(0.3, -0.45)$);
	\draw[line width=(#3)*0.5pt,scale=#3] ($(a)+(-0.3, -0.45)$) -- ($(a)+(0, -0.32)$);
	\draw[line width=(#3)*0.5pt,scale=#3] ($(a)+(0.3, -0.45)$) -- ($(a)+(0, -0.32)$);
	\draw[line width=(#3)*0.75pt,scale=#3] ($(a)+(0, -0.3)$) -- ($(a)+(-0.22, -0.17)$);
	\draw[line width=(#3)*0.75pt,scale=#3] ($(a)+(0, -0.3)$) -- ($(a)+(0.22, -0.17)$);
	\draw[line width=(#3)*0.5pt,scale=#3] ($(a)+(-0.22, -0.17)$) -- ($(a)+(0, -0.07)$);
	\draw[line width=(#3)*0.5pt,scale=#3] ($(a)+(0.22, -0.17)$) -- ($(a)+(0, -0.07)$);;
	\draw[line width=(#3)*0.75pt,scale=#3] (a) -- ($(a)+(-0.18, 0.11)$);
	\draw[line width=(#3)*0.75pt,scale=#3] (a) -- ($(a)+(0.18, 0.11)$);
	\draw[line width=(#3)*0.5pt,scale=#3] ($(a)+(-0.18, 0.11)$) -- ($(a)+(0,0.2)$);
	\draw[line width=(#3)*0.5pt,scale=#3] ($(a)+(0.18, 0.11)$) -- ($(a)+(0,0.2)$);
	\draw[line width=(#3)*0.5pt,scale=#3] ($(a)+(0, 0.3)$) -- ($(a)+(-0.1, 0.37)$);
	\draw[line width=(#3)*0.5pt,scale=#3] ($(a)+(0, 0.3)$) -- ($(a)+(0.1, 0.37)$);
	\draw[line width=(#3)*0.25pt,scale=#3] ($(a)+(-0.1, 0.37)$) -- ($(a)+(0, 0.43)$);
	\draw[line width=(#3)*0.25pt,scale=#3] ($(a)+(0.1, 0.37)$) -- ($(a)+(0, 0.43)$);
	\draw[line width=(#3)*0.75pt,scale=#3] ($(a)+(0, 1.2)$) -- ($(a)+(0,1)$);
	\draw[fill=white,scale=#3] ($(a)+(0, 1.2)$) circle (0.05cm);
	\draw[line width=(#3)*1pt,decorate,decoration=expanding waves,decoration={segment length=(#3)*10pt},scale=#3] ($(a)+(0, 1.2)$) -- ($(a)+(0, 2.5)$);
}
\ifCLASSINFOpdf
\else
\fi
\newcommand{\alert}[1]{\textcolor{black}{#1}}

\begin{document}
\title{Interference Mitigation via Rate-Splitting and Common Message Decoding in Cloud Radio Access Networks}
%\author{Alaa Alameer Ahmad, Hayssam Dahrouj, Anas Chaaban, Aydin Sezgin and Mohamed-Slim Alouini
%\thanks{}%  % <-this % stops a space
%}


\author{\IEEEauthorblockN{Alaa Alameer Ahmad\IEEEauthorrefmark{1}, \textit{Student Member, IEEE},
Hayssam Dahrouj\IEEEauthorrefmark{2}, \textit{Senior Member, IEEE},
Anas Chaaban \IEEEauthorrefmark{3},  \textit{Senior Member, IEEE},\\
Aydin Sezgin\IEEEauthorrefmark{1}, \textit{Senior Member, IEEE} and
Mohamed-Slim Alouini\IEEEauthorrefmark{4}, \textit{Fellow, IEEE}}
\thanks{Part of this paper was presented at the IEEE International Workshop on Signal Processing Advances in Wireless Communications (SPAWC), Kalamata, Greece, June 2018 \cite{8445778}.\newline A.~A.~Ahmed and A.~Sezgin are with the Department of Electrical Engineering Ruhr-university Bochum, Germany (Email: $\left\lbrace\text{alaa.alameerahmad, aydin.sezgin}\right\rbrace$@rub.de). H.~Dahrouj is with the Department of Electrical and Computer Engineering, Effat University, Saudi Arabia (Email: hayssam.dahrouj@kaust.edu.sa). A.~Chaaban is with the School of Engineering, The University of British Columbia, Kelowna, Canada (Email: anas.chaaban@ubc.ca). M.~S. Alouini is with the Communication Theory Lab, King Abdullah University of Science and Technology, Thuwal, Saudi Arabia (Email: slim.alouini@kaust.edu.sa).}
}

\maketitle

\begin{abstract}
Cloud-radio access networks (C-RAN) help overcoming the scarcity of radio resources by enabling dense deployment of base-stations (BSs), and connecting them to a central-processor (CP). This paper considers the downlink of a C-RAN, where the cloud is connected to the BSs via limited-capacity backhaul links. We propose and optimize a C-RAN transmission scheme that combines rate splitting, common message decoding, beamforming vectors design and clustering. To this end, the paper optimizes a transmission scheme that combines rate splitting (RS), common message decoding (CMD), clustering and coordinated beamforming. In this work we focus on maximizing the weighted sum-rate subject to per-BS backhaul capacity and transmit power constraints, so as to jointly determine the RS-CMD mode of transmission, the cluster of BSs serving private and common messages of each user, and the associated beamforming vectors of each user private and common messages. The paper proposes solving such a complicated non-convex optimization problem using $l_0$-norm relaxation techniques, followed by inner-convex approximations (ICA), so as to achieve stationary solutions to the relaxed non-convex problem. Numerical results show that the proposed method provides significant performance gain as compared to conventional interference mitigation techniques in C-RAN which simply treat interference as noise (TIN).\\

\end{abstract}
\section{Introduction}
\subsection{Overview}
%Device-to-device (D2D) communication is a networking paradigm in which devices may set up arbitrary peer-to-peer links to communicate with each other in a distributed fashion, in contrast to the traditional cellular paradigm that relies on the base-station infrastructure to coordinate transmissions. The localized transmissions in a D2D network provide benefits such as reduced latency, increased frequency reuse, and saving in device energy consumption. However, the design of D2D network must also address the central challenge of interference management. Because of the arbitrary locations of D2D links, a receiver in a D2D network may be located close to the senders of nearby links, consequently may suffer from much stronger interference as compared to that in a conventional cellular network, especially in a dense deployment. This work aims to devise distributed algorithms to manage such interference via scheduling and power optimization.
{\color{black} Motivated by the scarcity of radio resources and the ever increasing need for higher data rates and reliable wireless services, C-RAN provides a practical network architecture capable of boosting the spectral and energy efficiency in next generation wireless systems (5G and beyond) \cite{6824752, Le2015,6736747}. By connecting many BSs to the CP, C-RANs enable spatial reuse through dense deployment of small cells, and exploit the emerging cloud-computing technologies for managing large networks \cite{6923535, 7487951}.
	
	With ultra dense deployment of small cells, the distance between the base station (BS) and the end user decreases, which results in a better quality of the direct channel. This comes, however, at the cost of increasing inter-BS interference due to proximity of the BSs in neighbouring cells. Furthermore, in C-RAN, the performance of the system is also limited by the finite capacity of backhaul links, \cite{6786060, 7096298, 7105959, 7488289, 7499119,7593140,7925639}. Intuitively, in the extreme case when the backhaul capacity goes to infinity, the C-RAN is equivalent to a broadcast channel (BC). In the other extreme in which the backhaul links have zero-capacity, the C-RAN becomes equivalent to an interference channel (IC), the capacity of which is still a well-known open problem, even for the simple two-user IC, where treating interference as noise (TIN) is known to be a suboptimal strategy, especially in high-interference regimes \cite{Carleial78,1056307,4675741,7422790}. {With limited backhaul capacity, C-RAN bridges the two extremes}. {With this observation in mind, we investigate in this paper a transmission scheme which improves the performance of C-RAN in different regimes, i.e., in backhaul limited regimes and interference limited regimes.}
	
	% In rate splitting strategy initially introduced by \cite{Carleial78}, the message of each user is split into two parts, a private part decodable at the intended user only and a common part which can be decoded by another user. Such a strategy is shown to approach the capacity region of the IC in the seminal works of \cite{1056307, 4675741}.  
	In the rate splitting strategy, initially introduced by \cite{Carleial78} for the IC, the message of each user is split into two parts: a private part decodable at the intended user only and a common part which can be decoded by other user. Such a strategy is shown to
	approach the capacity region of the IC in the seminal works of \cite{1056307, 4675741}. 
	% In the seminal works on IC \cite{Carleial78, 1056307, 4675741}, it is shown that splitting the message of each user into two parts, a private part decodable at the intended user only and a common part which can be decoded by another user, can significantly approach the capacity region of the IC. 
	Motivated by this fact, this paper studies rate-spitting in the realm of a C-RAN. It proposes splitting the message of each user into two parts, a private part decodable at the intended user only, and a common part which can be decoded at a subset of users.
	
	Since the CP is connected to the BSs in cloud-enabled networks, C-RAN becomes a particularly suitable platform for the physical implementation of rate-splitting strategies. In the context of our paper, all rate splitting and common message decoding (RS-CMD) techniques are adopted for the sole purpose of reducing large-scale interference. As the CP is connected to the BSs via finite capacity backhaul links, it becomes equally important to determine the set of BSs (i.e., cluster) which serves each user, jointly with selecting the mode of transmission of each user (i.e., private, common, or both).
	
	This work considers the RS-CMD problem in the downlink of a C-RAN, where the CP is connected to several BSs, each equipped with multiple antennas. The CP applies central encoding to user's messages and establishes cooperation between a cluster of BSs by joint design of linear precoding in a user-centric clustering fashion, also known as \textit{data-sharing} strategy \cite{Simeone2009, 6588350, 5997324, 6920005}, as it achieves a better performance compared to classical transmission schemes \cite{7809154}.  The paper then considers the problem of maximizing the weighted sum-rate (WSR) across the network, subject to per-BS backhaul capacity and transmission power constraints. The goal of this optimization is to jointly determine the RS-CMD mode of transmission, the cluster of BSs serving private and common messages of each user, and the associated beamforming vectors of each user private and common information. The paper provides an in-depth numerical investigation of the impact of RS-CMD strategy on the achievable rate in C-RANs, and compares it with the conventional strategies which treat interference as noise.
	
	\subsection{Related Work}
	The contributions of this paper are related to works on rate splitting and common message decoding, clustering, and beamforming; topics which are studied in the literature of wireless systems, both individually and separately.
	
	In rate-splitting schemes, the data of each user is divided into two parts: a private message which is decoded only at the intended user, and a common message which is decodable at the intended user and a subset of the unintended users. Reference \cite{1056307} shows that such a RS-CMD technique leads to the largest known achievable rate-region in a 2-user IC. Such splitting strategy is further shown in \cite{4675741} to achieve rates within one-bit from the capacity of the 2-user IC. Although being based on simple networks, those information-theoretical studies show the benefits of using RS-CMD techniques in high interference regimes. For instance, inspired by the theoretical works in \cite{1056307, 4675741}, the authors in \cite{5910112} generalize this RS-CMD scheme to a practical multi-cell network showing significant achievable rate improvement by jointly designing the beamforming vectors for private and common information in RS-CMD as compared to beamforming design using TIN. In \cite{6125368}, the authors apply RS ideas to a practical setup of heterogeneous wireless networks. The results in \cite{6125368} suggest that a significant performance gain can be reached by applying RS as compared to rank-1 coordinated beamforming schemes that adopt TIN strategy. The work in \cite{7101259} uses common message decoding and successive interference cancellation techniques to maximize the sum rate in multi-cell multi-user MIMO system. The difference of convex optimization technique is used to efficiently solve the difficult underlying optimization problem.
	Recently, RS-CMD has also gained a noteworthy attention in the literature of medium access schemes. For instance, the authors in \cite{journals/corr/abs-1710-11018} propose a novel RS multiple access (RSMA) scheme, which generalizes and outperforms conventional multiple access schemes such as Space-Division Multiple Access (SDMA) and Non-Orthogonal Multiple Access (NOMA). Based on these results, the authors in \cite{8491100} show that RSMA is more energy efficient than SDMA and NOMA. Reference \cite{2018arXiv180205567M}, on the other hand, shows that linearly precoded RS is more efficient than the conventional Multi-User Linear Precoding (MU-LP) in terms of spectral and energy efficiency. Through numerical simulations, the authors in \cite{2018arXiv180205567M} particularly show that, with no increase in receiver complexity, RS achieves better performance metrics as compared to both NOMA and MU-LP systems.
	The above works, i.e., \cite{1056307, 4675741, 5910112,6125368,7101259,journals/corr/abs-1710-11018,8491100,2018arXiv180205567M }, however, do not address cloud-enabled scenarios, as they ignore the physical-layer considerations induced by RS-CMS in C-RANs, and do not account for determining the set of common messages. This paper, therefore, focuses on the study of the joint resource allocation problem in C-RAN, together with evaluating the impact of RS-CMD techniques. The paper further develops a well-chosen heuristic procedure to determine the set of common messages that each user needs to decode.
	
	In general, most of the existing works (e.g., \cite{4443878,5463229,5594708,5935083,5962539,6512541}) on multi-cell interference mitigation in practical networks focus on doing so through jointly allocating resources (e.g., beamforming vectors and transmit power) in order to maximize a network utility. References \cite{4443878,5463229,5594708,5935083,5962539,6512541}, however, often adopt the strategy of TIN and assume an infinite backhaul capacity. Towards this end, the impact of finite backhaul links capacity is studied in the downlink of C-RAN in \cite{6920005}. The problem studied in \cite{6920005} turns out to be a mixed-integer non linear problem (MINLP), which is solved by relaxing the discrete non-convex per-BS backhaul constraints using re-weighted $l_1$-norm, and then by applying a generalized weighted minimum mean square algorithm (WMMSE). The authors in \cite{6786060} consider the joint design of BSs' clusters and beamforming vectors to minimize the network-wide transmit power cost. The trade-off between the backhaul traffic and transmit power is also investigated in references \cite{6920005, 7105959,7488289,7499119,7593140, 7925639}, all of which adopt TIN to decode the received messages. At this point, it becomes essential to investigate how adopting RS-CMD can influence the design of clusters of BSs and the beamforming vectors associated with the private and common messages in a C-RAN setup. Towards this end, our current paper investigates the downlink C-RAN by utilizing a RS-CMD strategy, and focuses on evaluating its impact on jointly optimizing the beamforming vectors, the clustering and the transmission mode, so as to maximize the weighted-sum rate (WSR) across the network. To the best of authors' knowledge, this is the first work on C-RAN which studies both the application of RS-CMD coupled with joint clustering and beamforming, and numerically illustrates the potential gain provided by RS-CMD over TIN.
\subsection{Contributions}
In this paper, we propose using RS-CMD in downlink C-RAN to jointly design user centric clusters of BSs, so as to explore RS-CMD benefits in large-scale interference management. We formulate a WSR maximization problem subject to per-BS backhaul capacity and per-BS transmit power constraints, so as to determine the RS splitting mode, the cluster of BSs which serves each user, and the beamforming vectors associated with the private and common messages parts. Such a problem is generally NP-hard due to its mixed discrete and continuous optimization nature, in additional to the non-convexity of the constraints. Our paper proposes solving such a problem using a heuristic based on $l_0$-norm approximation to tackle the discrete part, followed by a polynomial time algorithm based on inner convex approximations, so as to find a stationary solution to the resulting non-convex continuous problem. The paper subsequently shows the numerical benefits of the proposed RS-CMD scheme in improving the achievable rates in C-RAN compared to the state-of-the art TIN strategy, both in the backhaul-limited and in the interference-limited regimes. Our main contributions can be summarized as follows:
\begin{itemize}
	\item \textit{Common Message Decoding (CMD) Set}:
	We propose a heuristic procedure for { ordering} the set of strongest interferers for each user, which consequently allows for determining the set of common messages to be decoded.
	\item \textit{Clustering}: Based on $l_0$-norm relaxation and inner-convex approximation framework, we propose a dynamic clustering approach. In the context of RS-CMD, we determine the set of BSs serving the private message and the set of BSs which serves the common message for each scheduled user. As opposed to the static clustering scheme described in \cite{6920005}, dynamic clustering procedure forms the clusters by taking into account the CMD set of each user, which can significantly affect the network connectivity. {To deal with the non-convex backhaul constraint, the paper particularly proposes a surrogate convex function to approximate the backhaul constraint. The paper then compensates for such approximations using proper outer-loop updates in an iterative manner.}
	\item \textit{Beamforming}: Even when the clusters are fixed, the WSR problem with RS-CMD is not convex. This is because the private and common rate functions are non-convex in the private and common beamforming vectors, respectively. The paper, therefore, proposes solving such issue using an algorithm that applies well-chosen inner-convex approximations. The proposed algorithm is proven to converge in polynomial time to a stationary solution.
	\item \textit{Numerical Simulations}: We show through extensive numerical simulations that our proposed solution outperforms the classical TIN in C-RAN. In both the interference limited and the backhaul limited regimes, we illustrate that RS-CMD makes a better use of the network resources in order to achieve higher rates as compared with TIN for different network parameters.
\end{itemize}
The rest of the paper is organized as follows. Section II illustrates the system model. Section III introduces the transmission scheme adopted in this work and formulates the WSR problem accordingly. The proposed solution is introduced in section IV. Section V presents the numerical simulations, and section VI concludes the paper.

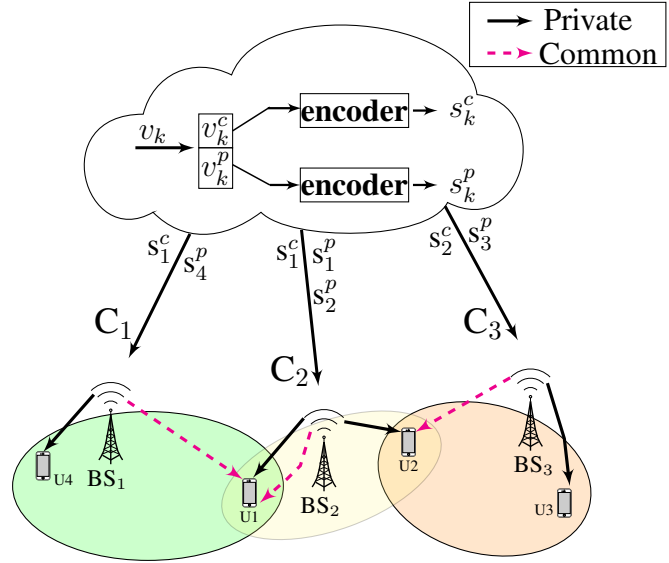
\begin{figure}[t]
	\centering
	\tikzset{every picture/.style={scale=.9}, every node/.style={scale=.9}}

\begin{tikzpicture}
    \tikzstyle{every node}=[font=\tiny]
    \tikzset{
        hexagon/.style={
            regular polygon,
            regular polygon sides=6,
            minimum size=23mm,
            inner sep=0mm,
            outer sep=0mm,
            rotate=0,
            draw
        }

    }

    % Ellipse
    \draw [black,fill=orange!20, rotate = -15] (2.8,-3.8) ellipse (1.6cm and 1.1cm);
    \draw [black,fill=green!20] (-3.25,-4.5) ellipse (2cm and 1.1cm);
    \draw [black,fill=yellow!60, rotate around={110:(0.75,-3.5)}, opacity=.2] (0.4,-2.0) ellipse (0.8cm and 1.7cm);

    % BS -- MU
    \draw [->, very thick] (-4.05,-3.15) -- (-4.8,-4.0); % Bs1 -- u4
    \draw [->, very thick, dashed,magenta] (-3.55,-3.25) -- (-1.85,-4.4); % bs1 -- u1

    \draw [->, very thick] (-0.95,-3.5) -- (-1.7,-4.38); % bs2 -- u1
    \draw [->, very thick, dashed,magenta] (-0.85,-3.7) -- (-1.0,-4.2) -- (-1.2,-4.4) -- (-1.6,-4.75); % bs2 -- u1
    \draw [->, very thick] (-0.35,-3.55) -- (0.5,-3.7); % bs2 -- u2

    \draw [->, very thick, dashed,magenta] (2.1,-2.9) -- (0.7,-3.7); %bs3 -- u2
    \draw [->, very thick] (2.65,-2.98) -- (2.9,-3.8) -- (2.98,-4.5); %bs3 -- u3

    % BS
    % Old Notes
    \node (bs11) at (-5,-4.6){};
    \node (bs31) at (3.2,-3.4){};
    \node (bs21) at (-0.5,-5){};

    \basestation{-3.8}{-3.8}{0.35};
    \node (bs1) at (-3.85,-4.4) {\footnotesize BS$_1$};

    \basestation{-0.65}{-4.2}{0.35};
    \node (bs2) at (-0.7,-4.8) {\footnotesize BS$_2$};

    \basestation{2.4}{-3.6}{0.35};
    \node (bs3) at (2.45,-4.2) {\footnotesize BS$_3$};

    %cloud
    % Cloud
    \node[cloud,white, draw,fill=white, minimum width=6.45cm , minimum height=3.2cm, cloud puffs= 12](cloud) at (-1,1){};

    \node at (-1,0.85)(ccc)
    {\begin{tikzpicture}
        \draw[scale = 1.45] (-2.1,-0.7) .. controls (-2.8,-1.1)
        and (-3.2,0.3) .. (-2.2,0.3) .. controls (-2.1,0.7)
        and (-1.7,0.9) .. (-1.3,0.7) .. controls (-1.0,1.5)
        and (1.1,1.3) .. (1.2,0.5) .. controls (2.0,0.4)
        and (1.7,-1) .. (0.9,-0.6) .. controls (0.7,-1)
        and (-0.7,-1) .. (-1.0,-0.7) .. controls (-1.4,-1)
        and (-1.8,-1) .. cycle;
        \end{tikzpicture}};

    \node (c1) at ($(cloud.240)!0.57!(bs11)$) {};
    \node (c2) at ($(cloud.270)!0.57!(bs21)$) {};
    \node (c3) at ($(cloud.310)!0.7!(bs31)$) {};
    %

    %\draw[draw,solid,line width=3.7mm,fill,
    %preaction={-triangle 90,ultra thick,draw,shorten >=-2mm}] (cloud.232) -- (c1);
    %\draw[draw=white,solid,line width=3mm,fill=white,
    %preaction={-triangle 90,thick,draw=white,shorten >=-1.5mm}](cloud.232) -- (c1);
    \node (temp1) at (-2.55,-0.61) {};
    \draw[->, line width=0.4mm] (temp1) to (c1);%cloud.232

    %\draw[draw,solid,line width=3.7mm,fill,
    %preaction={-triangle 90,ultra thick,draw,shorten >=-2mm}] (cloud.270) -- (c2);
    %\draw[draw=white,solid,line width=3mm,fill=white,
    %preaction={-triangle 90,thick,draw=white,shorten >=-1.5mm}](cloud.270) -- (c2);
    \node (temp2) at (-1.0,-0.55) {};
    \draw[->, line width=0.4mm] (temp2) to (c2);%cloud.270

    %\draw[draw,solid,line width=3.7mm,fill,
    %preaction={-triangle 90,ultra thick,draw,shorten >=-2mm}] (cloud.312) -- (c3);
    %\draw[draw=white,solid,line width=3mm,fill=white,
    %preaction={-triangle 90,thick,draw=white,shorten >=-1.5mm}](cloud.312) -- (c3);
    \node (temp3) at (1.05,-0.2) {};
    \draw[->, line width=0.4mm] (temp3) to (c3);%cloud.312

    % Cloud inner parts
    %\node[cloud, draw,fill=white, minimum width=5.75cm , minimum height=3cm, cloud puffs= 12](cloud) at (-1,0.5){};
    \draw (-2.5,.5) rectangle ++(.5,.5) node[pos=.5] {\normalsize $v_{k}^c$};
    \draw (-2.5,-.1) rectangle ++(.5,.6) node[pos=.5] {\normalsize $v_{k}^p$};
    %\draw[line width=0.5mm] (-2.5,0.55)--(-2,0.55);
    \node  at (-3.6,0.5) (h1){};
    \node  at (-2.4,0.5) (h2){};
    \node (wer) at (-2.8,0.7) [anchor=east] {\normalsize $v_k$};
    \draw[->, line width=0.4mm] (h1) to (h2);
    \draw[line width=0.2mm] (-2,0.75) --  ++(0.55,0.35);
    \draw[line width=0.2mm] (-2,0.25) --  ++(0.55,-0.3);
    \draw[->,line width=0.3mm] (-1.45,-.05) --  ++(0.45,0);
    \draw (-1,0.8) rectangle ++(1.6,0.5) node[pos=.5] {{\normalsize \textbf{encoder}}};
    \draw (-1,-.3) rectangle ++(1.6,0.5) node[pos=.5] {{\normalsize \textbf{encoder}}};
    \node  at (0.5,-.05) (h3){};
    \node[align=right]  at (1.4,-.05) (h4){\normalsize \bm $s_{k}^p$};
    \draw[->, line width=0.3mm] (h3) to (h4);
    \node  at (0.5,1.05) (h5){};
    \node  at (1.4,1.05) (h6){\normalsize \bm $s_{k}^c$};
    \draw[->, line width=0.3mm] (h5) to (h6);
    \draw[->,line width=0.3mm] (-1.48,1.085) --  ++(0.5,0);
    %\node[align=right] at(){};

    %labenodes
    %\node [align=right] (c1label1) at ($(cloud.242)!0.15!(c1)$) {$\text{s}_{2}^p$};%235
    %\node [align=left] (c1label2) at ($(cloud.226)!0.05!(c1)$) {$\text{s}_{2}^c$};
    \node [align=left] (c1label3) at ($(cloud.218)!0.215!(c1)$) {\normalsize$\text{s}_{1}^c$};
    \node [align=right] (c1label3) at ($(cloud.235)!0.39!(c1)$) {\normalsize$\text{s}_{4}^p$};

    \node [align=right] (c2label1) at ($(cloud.282)!0.15!(c2)$) {\normalsize$\text{s}_{1}^p$};
    \node [align=right] (c2label2) at ($(cloud.262)!0.13!(c2)$) {\normalsize$\text{s}_{1}^c$};
    %\node [align=right] (c2label3) at ($(cloud.262)!0.3!(c2)$) {$\text{s}_{2}^c$};
    \node [align=right] (c2label4) at ($(cloud.284)!0.4!(c2)$) {\normalsize$\text{s}_{2}^p$};

    \node [align=right] (c3label1) at ($(cloud.336)!0.29!(c3)$) {\normalsize$\text{s}_{3}^p$};
    \node [align=right] (c3label1) at ($(cloud.320)!0.17!(c3)$) {\normalsize$\text{s}_{2}^c$};

    \node[align=center] (c3temp) at (1.7,-2.1) {\large $\text{C}_3$};

    \node[align=center] (c1temp) at (-3.75,-2.1) {\large $\text{C}_1$};

    \node[align=right] (c2temp) at (-1.15,-2.8) {\large $\text{C}_2$};

    %% Phones
    \node at (-1.75,-4.6)(phone1)
    {\begin{tikzpicture}
        \iPhone{0}{1}{0.3};
        \end{tikzpicture}};
    \node at (-1.75,-4.95) {U1};

    \node at (0.6,-3.85)(phone2)
    {\begin{tikzpicture}
        \iPhone{0}{1}{0.3};
        \end{tikzpicture}};
    \node at (0.6,-4.2) {U2};

    \node at (-4.8,-4.2)(phone4)
    {\begin{tikzpicture}
        \iPhone{0}{1}{0.3};
        \end{tikzpicture}};
    \node at (-4.5,-4.4) {U4};

    \node at (2.9,-4.75)(phone5)
    {\begin{tikzpicture}
        \iPhone{0}{1}{0.3};
        \end{tikzpicture}};
    \node at (2.6,-4.85) {U3};

    %helpstuff%
    %\draw[red] (-5,-5) grid (5,2);
    \node  at (1.6,1.9) (a1){};
    \node  at (2.6,1.9) (a2){};
    \node  at (1.6,2.375) (b1){};
    \node  at (2.6,2.375) (b2){};
    % Legend
    \draw [->, very thick, dashed,magenta] (a1) to (a2);\node[align=left] at (3.4,1.9) {\normalsize Common};
    \draw [->, very thick] (b1) to (b2);\node[align=left] at (3.25,2.4) {\normalsize Private};
    \draw[draw=black] (1.5,1.7) rectangle (4.4,2.65);

\end{tikzpicture}
\caption{A C-RAN system with three cells. Both private and common messages are designed at the cloud.}
\label{Model}
\end{figure}
\section{System Model}
We consider a C-RAN system operating in downlink mode with a transmission bandwidth $B$. The network consists of a set of multi-antenna BSs $\mathcal{N} = \left\lbrace1,2,\ldots,N\right\rbrace$, serving a set of single-antenna users $\mathcal{K} = \left\lbrace1,2,\ldots,K\right\rbrace$. Each BS is equipped with $L\geq 1$ antennas. BS $n \in \mathcal{N}$ is connected to a CP, located at the cloud, via a backhaul link of capacity $C_n$. User $k$ requires a message $v_k$, where the achievable data-rate at user $k$ is denoted by $R_k$. All messages are jointly encoded at the CP into signals $s_k$, $\forall k\in \mathcal{K}$. The CP then shares combinations of $s_k$ (or parts thereof) with the BSs through the backhaul links. This data-sharing is possible if the rate of signals shared with BS $n$ does not exceed the backhaul capacity $C_n$. This is made more explicit when we describe RS in the next section.
\alert{Upon receiving these signals, BS $n$ constructs} $\mathbf{x}_{n} \in \mathbb{C}^{L\times1}$, and sends it according to the following transmit power constraint:

\begin{equation}\label{eq:Po}
\EX\left\lbrace \mathbf{x}_{n}^{H}\mathbf{x}_{n}\right\rbrace \leq P_{n}^{\text{Max}} \quad \forall n \in \mathcal{N},
\end{equation}
where $P_{n}^{\text{Max}}$ is the maximum transmit power available at BS $n$.

Let $\mathbf{h}_{n, k} \in \mathbb{C}^{L\times1}$ denote the channel vector between BS $n$ and user $k$, and $\mathbf{h}_k = \big[\mathbf{h}_{1, k}^{T}, \mathbf{h}_{2, k}^{T},...,\mathbf{h}_{N, k}^{T}\big]^{T}  \in \mathbb{C}^{NL \times 1}$ be the aggregate channel vector of user $k$. We can write the received signal at user $k$ as
\begin{equation}
y_k = \mathbf{h}_{k}^H \mathbf{x} + n_k
\end{equation}
where $n_k \sim \mathcal{CN}\left(0, \sigma^2\right)$ is the additive white Gaussian noise (AWGN), and $\mathbf{x} = [\mathbf{x}_{1}^{T},\ldots,\mathbf{x}_{N}^{T}]^{T}$.
%which is independent of the noise at the receivers $m, \hspace{0.2mm} m \neq k$ and of the user's data and

For mathematical tractability, the paper assumes that the CP has complete knowledge of the instantaneous channel state information (CSI) of all BSs. We further adopt a block-based transmission model, where each transmission block consists of several time slots. The channel fading coefficients remain constant within one block, but may vary independently from one block to another.
Next, we describe our proposed scheme which is based on RS-CMD, and we formulate the WSR optimization problem accordingly.

\section{Transmission Scheme and Problem Formulation}
The proposed transmission scheme consists of RS, joint beamforming and data-sharing, and successive common message decoding. We start by describing RS.

\subsection{Rate Splitting}
The CP first splits the message of user $k$, i.e., $v_k$, into a private message denoted by $v_{k}^p$, and a common message denoted by $v_{k}^c$. Afterwards, the CP encodes the private and common messages into $s_{k}^p$ and $s_{k}^c$, respectively, as illustrated in Fig. \ref{Model}. The coded messages $s_{k}^p$ and $s_{k}^c$ are assumed to be i.i.d. circularly symmetric complex Gaussian with zero mean and unit variance. Their respective rates are denoted by $R_k^p$ and $R_k^c$, and so $R_k=R_k^p+R_k^c$, where $R_k$ is the rate of user $k$.

\subsection{Beamforming, Signal Construction, and Data-Sharing}
{When adopting the data-sharing strategy in a downlink mode in C-RAN which applies RS-CMD, the CP shares the encoded private and common messages directly with their respective cluster of BSs.
}
Let $\mathcal{K}_n^p,\mathcal{K}_n^c \subseteq \mathcal K$ be the subset of users served by BS $n$ with a private or common message, respectively, i.e.,
\begin{align}
\mathcal{K}_{n}^{p} &\defeq \left\lbrace k \in \mathcal{K}|\hspace{1mm} \text{BS \textit{n} delivers}\hspace{0.5mm} s_{k}^p \hspace{1mm} \text{to user}\hspace{0.5mm} k\right\rbrace,\\
\mathcal{K}_{n}^{c} &\defeq \left\lbrace k \in \mathcal{K}|\hspace{1mm} \text{BS \textit{n} delivers}\hspace{0.5mm} s_{k}^c \hspace{1mm} \text{to user}\hspace{0.5mm} k\right\rbrace.
\end{align}
Moreover, let the beamformers used by BS $n$ to send $s_k^p$ and $s_k^c$ to user $k$ be denoted by $\mathbf{w}_{n,k}^p$ and $\mathbf{w}_{n,k}^c$, respectively. \alert{Then, the CP sends $\left\lbrace {s}_{k}^p| \forall k \in \mathcal{K}_{n}^{p}\right\rbrace $, $\left\lbrace {s}_{k}^c| \forall k \in \mathcal{K}_{n}^{c}\right\rbrace $ and their beamforming vectors over the backhaul links to BS $n$. Due to the finite backaul capacity $C_n$ limits, the transmission rate is subject to the following backhaul capacity constraint\footnote{\alert{We ignore the overhead due to sending the beamformers since these need to be sent only when CSI changes}.}:
	\begin{equation} \label{eq:S4}
	\sum\limits_{k \in \mathcal{K}_{n}^p} R_{k}^p + \sum_{k \in \mathcal{K}_{n}^c} R_{k}^c\leq C_n, \quad \forall n\in \mathcal{N}
	\end{equation}
	BS $n$ then constructs $\mathbf{x}_n$ as follows:}
\begin{equation}\label{eq:Tr}
\mathbf{x}_n = \sum_{k\in\mathcal{K}_n^p}\mathbf{w}_{n,k}^p {s}_{k}^p + \sum_{k\in\mathcal{K}_n^c}\mathbf{w}_{n,k}^c {s}_{k}^c.
\end{equation}

Using the expression of the transmit signal \eqref{eq:Tr}, one can rewrite the power constraint \eqref{eq:Po} as follows:
\begin{equation}\label{eq:S}
\sum\limits_{k \in \mathcal{K}}\left({\normV[\big]{\mathbf{w}_{n, k}^p}_{2}^{2}} + {\normV[\big]{\mathbf{w}_{n, k}^c}_{2}^{2}}\right) \leq P_{n}^{\text{Max}}, \quad \forall n\in \mathcal{N}.
\end{equation}

The private (common) message of user $k$ is served by BS $n$, if the corresponding beamforming vector $\mathbf{w}_{n,k}^p$ ($\mathbf{w}_{n,k}^c$) is non-zero. This can be equivalently expressed in terms of the indicator function as follows:
%$\normV[\Big]{$
\begin{equation}
\mathds{1}\left\lbrace \normV[\big]{\mathbf{w}_{n, k}^o}_{2}^{2}\right\rbrace =  \begin{cases}
1 \quad \text{if} \quad &\normV[\big]{\mathbf{w}_{n, k}^o}_{2}^{2} > 0 \\
0 \quad  \quad &\text{otherwise}
\end{cases}
\end{equation}
where, $o \in \left\lbrace p, c \right\rbrace $.
Without loss of generality, the above indicator function can be written as a function of an $l_0$-norm notation \footnote{$l_0$-norm of a vector is the number of non-zero elements in this vector. } , i.e., as $\mathds{1}\left\lbrace \normV[\big]{\mathbf{w}_{n, k}^o}_{2}^{2} \right\rbrace =  \normV[\Big]{\normV[\big]{\mathbf{w}_{n, k}^o}_{2}^{2}}_0$. This is this case since, in the scalar case, the $l_0$-norm definition coincides with the definition of the indicator function, because the power transmitted from BS $n$ to user $k$ is a positive scalar, i.e, $\normV[\big]{\mathbf{w}_{n, k}^o}_{2}^{2} \in \mathbb{R}_{+}$.
The subset of users served with private and common messages from BS $n$ can, therefore, be expressed as:
\begin{align}
\mathcal{K}_{n}^{p} & = \left\lbrace k|\quad \normV[\Big]{\normV[\big]{\mathbf{w}_{n, k}^p}_{2}^{2}}_0 = 1 \right\rbrace,\\
\mathcal{K}_{n}^{c} & =  \left\lbrace k|\quad \normV[\Big]{\normV[\big]{\mathbf{w}_{n, k}^c}_{2}^{2}}_0 = 1 \right\rbrace.
\end{align}
The above expressions allow to re-express the backhaul constraint \eqref{eq:S4} in the following compact form:
\begin{equation} \label{eq:S11}
\sum\limits_{k \in \mathcal{K}}\left(\normV[\Big]{\normV[\big]{\mathbf{w}_{n, k}^p}_{2}^{2}}_0 R_{k}^p + \normV[\Big]{\normV[\big]{\mathbf{w}_{n, k}^c}_{2}^{2}}_0 R_{k}^c \right) \leq C_n, \quad \forall n\in \mathcal{N}.
\end{equation}
\subsection{Successive Decoding}
At this step, the received signal at user $k$ can be written as
\begin{align*}
y_k =& \mathbf{h}_{k}^H\left(\mathbf{w}_{k}^p {s}_{k}^p + \mathbf{w}_{k}^c {s}_{k}^c\right) + \sum_{j \in \mathcal{K}\setminus \{k\}}\mathbf{h}_{k}^H\left(\mathbf{w}_{j}^p {s}_{j}^p + \mathbf{w}_{j}^c {s}_{j}^c\right)\\& + n_k,
\end{align*}
where $\mathbf{w}_{k}^p = [(\mathbf{w}_{1, k}^{p})^{T},...,(\mathbf{w}_{N, k}^{p})^{T}]^T$ is the aggregate beamforming vector associated with ${s}_{k}^p$, i.e., the private message of user $k$. Similarly, $\mathbf{w}_{k}^c$ is the aggregate beamforming vector associated with ${s}_{k}^c$, i.e., the common message of user $k$.

In the context of this paper, using common messages is adopted for the sole purpose of mitigating interference in C-RANs. Thus, the order in which user $k$ decodes the intended messages plays an important role in assessing the efficiency of the relevant proposed interference mitigation techniques. Although joint decoding of all common and private messages at user $k$ would result in optimized rates, its implementation is complicated in practice, in particular when the network and the intended set of messages to be decoded by each user are large. The classical information theoretical results of a 2-user IC, however, already suggest that decoding a strong interferer's common message can significantly improve a user's achievable rate \cite{4675741}. From this perspective, in this paper, we focus on a successive decoding strategy, wherein user $k$ decodes a subset of all common messages in a fixed decoding strategy, based on the descending order of the channel gains of the interferers, as described next.

Let $\mathcal{M}_k$ denote the set of users which decode $s_k^c$, i.e.:
\begin{equation}
\mathcal{M}_k \defeq \left\lbrace j \in \mathcal{K}| \hspace{1mm}\text{user} \hspace{1mm}j\hspace{1mm}  \text{decodes} \hspace{1mm}s_{k}^c \right\rbrace.
\end{equation}
The set of common messages that user $k$ would decode is then defined as:
\begin{align}
\Phi_{k} \defeq \left\lbrace j \in \mathcal{K}|\hspace{1mm} k \in \mathcal{M}_{j} \right\rbrace.
\end{align}
We note that once the set $\mathcal{M}_k$ is found, we can determine the set $\Phi_{k}$, and vice-versa.
The choice of $\Phi_{k}$ (and consequently $\mathcal{M}_k$) has a crucial impact on the achievable rate of user $k$. In this paper, we design $\Phi_{k}$ (and $\mathcal{M}_k$) in a heuristic fashion, which is based on the order of the interfering channel gains.
%which is justified by the recent optimality results of TIN as reported in \cite{6544312, 7486985,6875098 }.
%and is computed as outlined in Algorithm \ref{Alg1}.

Consider the following decoding order at user $k$:
{\color{black}
	\begin{equation}
	\pi_{k}(j): \left\lbrace1, 2, \ldots,\left|\Phi_{k}\right|  \right\rbrace  \rightarrow  \Phi_{k} \nonumber,
	\end{equation}
	which represents a permutation of an ordered set with cardinality of $\left|\Phi_{k}\right|$,}
i.e., $\pi_{k}(j)$ is the successive decoding step in which the message $j \in \Phi_k$ is decoded at user $k$. { In other terms, $\pi_{k}(j_1) > \pi_{k}(j_2)$ (where $j_1 \neq j_2$) implies that user $k$ decodes the common message
	of user $j_1$ first, and then the common message of user user $j_2$.}
Now, write $\mathbf{y}_k$, the received signal at user $k$, as follows,
\begin{align}\label{eq:S2}
y_k = & \underbrace{\left(\mathbf{h}_{k}^H\mathbf{w}_{k}^p {s}_{k}^p + \sum_{j \in \Phi_{k}}\mathbf{h}_{k}^H\mathbf{w}_{j}^c {s}_{j}^c\right)}_{\text{Signals to be decoded}} \nonumber \\ & + \underbrace{\sum_{j \in \mathcal{K}\setminus k}\mathbf{h}_{k}^H\mathbf{w}_{j}^p {s}_{j}^p + \sum_{l \in \mathcal{K} \setminus \Phi_{k}}\mathbf{h}_{k}^H\mathbf{w}_{l}^c {s}_{l}^c+ n_k.}_\text{Interference plus noise}
\end{align}
{ Since finding the optimal decoding order is obviously a challenging problem for its combinatorial nature, we herein propose a practical successive decoding strategy instead. The idea is to fix the decoding order according to channel strength in descending order as follows: $\normV[]{\mathbf{h}_{\pi_{k}(1)}}\geq\normV[]{\mathbf{h}_{\pi_{k}(2)}} \geq \ldots \geq \normV[]{\mathbf{h}_{\pi_{k}(\left|\Phi_{k}\right|)}}$. Such decoding strategy helps the users whose common messages are decoded achieving better common rates. Although the proposed decoding technique does not provide the global optimal solution to the problem, the simulations section of the paper later illustrate how that such a decoding order indeed provides an appreciable gain as compared to the conventional private-information transmission only, i.e., TIN.}
\subsection{Achievable Rate}
Let $\Gamma_{k}^p, \Gamma_{k,i}^c$ denote the signal to interference plus noise ratios (SINR's) of user $k$, when decoding its private message and the common message of user $i$, respectively. Based on equation \eqref{eq:S2}, we can write:
\begin{align}
\label{eq:S3}
\Gamma_{k}^p &= \frac{\left|\mathbf{h}_{k}^H\mathbf{w}_{k}^p \right|^2}{\sum\limits_{j \in \mathcal{K}\setminus k}\left|\mathbf{h}_{k}^H\mathbf{w}_{j}^p \right|^2 + \sum\limits_{l \in \mathcal{K} \setminus \Phi_{k}}\left|\mathbf{h}_{k}^H\mathbf{w}_{l}^c \right|^2 + \sigma^2}\\
\label{eq:S6}
\Gamma_{k, i}^c &= \frac{\left|\mathbf{h}_{k}^H\mathbf{w}_{i}^c \right|^2}{T_k + \sum\limits_{l \in \mathcal{K} \setminus \Phi_{k}}\left|\mathbf{h}_{k}^H\mathbf{w}_{l}^c \right|^2 + \sum\limits_{\substack{m  \in \Phi_{k}\\ \pi_{k}(m)> \pi_{k}(i)}}\left|\mathbf{h}_{k}^H\mathbf{w}_{m}^c \right|^2}
\end{align}
where $T_k = \sum_{j \in \mathcal{K}}\left|\mathbf{h}_{k}^H\mathbf{w}_{j}^p \right|^2 + \sigma^2$. The above expressions (\ref{eq:S3}) and (\ref{eq:S6}) assume that each user decodes its private message last, which is adopted for its capability to reduce the interference through common message decoding, as in the classical multi-cell systems \cite{5910112}.} The total achievable rate of user $k$, $R_k = R_{k}^p + R_{k}^c$, then satisfies the following achievability conditions:
\begin{align}
\Gamma_{k}^p &\geq 2^{R_{k}^p/B} - 1, \quad \forall k \in \mathcal{K}, \label{eq:S7}\\
\Gamma_{i,k}^c &\geq 2^{R_{k}^c/B} - 1, \quad \forall i \in \mathcal{M}_k \hspace{1.5mm} \text{and} \hspace{1.5mm} \forall k \in \mathcal{K}. \label{eq:S8}
\end{align}
\subsection{Determining the Common Message Sets}
The latest results of TIN in interference networks, e.g., \cite{6875098}, suggest a scheduling procedure to manage interfering links in a device-to-device (D2D) network. The idea in \cite{6875098} is to allow the links which meet the TIN optimality criteria to share the same resources block (bandwidth, transmit frequency). Optimality of TIN criteria is then illustrated in terms of generalized degrees-of-freedom. In short, if a link causes much interference to other links (already scheduled to a transmitting resource block), or suffers from much interference, then one should schedule it to another block.\\
In the context of our paper, instead of scheduling users to other transmitting blocks, we propose to deploy RS-CMD strategy for the users which cause high levels of interference to other users, so as to determine a heuristic, yet reasonable, strategy for determining the common message sets. To this end, we propose a simple criterion to identify the users which receive too much interference (weak users), and allow them to decode the common messages of strong interferers (strong users). The network we are interested in is more complex than those studied in \cite{6544312, 7486985,6875098 }. The proposed criterion, although being a heuristic one, leads to a significant gain over the TIN strategy used in the state-of-the art C-RAN, as illustrated later in the simulations section.\\
Our proposed algorithm relies on first identifying the users for which TIN is not optimal, i.e., solely based on their channel gains. We do so by initializing the beamformers of all users as feasible maximum ratio combining (MRC) beamformers. Then we compute the achievable rates, and for each user, we evaluate the total interference received from other users. {To best identify whether a user is considered as a weak or a strong interferer, we define a parameter $\mu$ as a separating threshold. More specifically, if the rate of a user $k$ is within the $\mu$th percentile, the user is considered a weak user, and up to ${D}$ strongest interferes of user $k$ are added to the set $\Phi_k$. {\color{black} Here, $D$ represents the number of layers in successive decoding strategy}. We note that $\mu$ plays an important role in bridging the gap between RS with RS-CMD. In other terms, when $\mu$ is small, only the weakest users would decode the common message of their interferers. {\color{black}By increasing $\mu$, however, more users participate in decoding the common messages of their interferers. The value of $\mu$ plays an important role in determining the gain of RS-CMD over TIN as the simulations results later suggest}}.\\
The above strategy guarantees that user $k$ would mitigate the interference it receives by decoding the common message of the strongest interferer.
The intuition behind this is that, if the rate of a user $k$ is high relative to other weakest users, this user would not be receiving a high level of interference, which makes it less useful that user $k$ would decode the common message of other users. The steps of determining the set of common messages for all users $k \in \mathcal{K}$ are summarized in Algorithm 1 description below.
\begin{algorithm}[h]
\caption{Procedure to Identify $\left\lbrace \Phi_k \right\rbrace_{k =1}^K $}
\begin{algorithmic}[1]
\State \textbf{Input}: CSI matrix $\mathbf{{H}}$, set of active users $\mathcal{K}$ and initialize $\left\lbrace \Phi_k = \left\lbrace k\right\rbrace  \right\rbrace_{k =1}^K $.
\State Compute the beamformers as $\mathbf{{W}} = \mathbf{{H}}^H$. % in \cite{6832894}.
\State Compute the achievable rates using TIN, based on step 2,.
\For{$k \in \mathcal{K}$}
\State $\widehat{\mathcal{K}} \leftarrow \mathcal{K}\setminus \{k\}$
%{\For{$j = 1:L$}}
\State Compute the interference power $\left\lbrace \text{\sf I}_{k, i}\right\rbrace_{i \in \widehat{\mathcal{K}}} $ as observed at user $k$.  \\
%\State Compute the interference of user $k$ caused at other active users  $\left\lbrace \text{IN}_{i, k}\right\rbrace_{i \in \mathcal{K}\setminus k} $.
\If{$R_k$ is within the $\mu$-th percentile of other users rate}
\State $\Phi_k = \Phi_k\cup\left\lbrace \underset{i \in \widehat{\mathcal{K}}}{\text{argmax}}\hspace{1mm} \text{\sf I}_{k, i} \right\rbrace $
\State $\widehat{\mathcal{K}} \leftarrow \widehat{\mathcal{K}} \setminus \left\lbrace\underset{i \in \widehat{\mathcal{K}}}{\text{argmax}}\hspace{1mm} \text{\sf I}_{k, i} \right\rbrace$
\If{$\left|\Phi_{k}\right| > L$}
\State $\mathcal{K} \leftarrow \mathcal{K}\setminus \left\lbrace k \right\rbrace $
\EndIf
\EndIf
\EndFor
\end{algorithmic}
\label{Alg1}
\end{algorithm}

%Note that according to algorithm 1, each user can decode only one additional common message besides its own common message. However, a common message of a user can be decoded by multiple users. This helps reducing the complexity of the receiver, by limiting the number of successive cancellation stages.
%\section{Problem Formulation and Proposed Solution}
\subsection{Problem Formulation}
{The optimization problem considered in this paper focuses on maximizing the weighted sum-rate (WSR) in RS-CMD C-RAN. The goal is to determine the common and private beamformers jointly with the common and private clusters of BSs associated with each user, subject to per BS transmission power and backhaul constraints. The considered WSR maximization problem can be mathematically written as:
	\begin{subequations}\label{eq:Opt1}
		\begin{align}
		&\underset{\left\lbrace \mathbf{w}_{k}^p, \mathbf{w}_{k}^c| \forall k \in \mathcal{K}\right\rbrace}{\text{maximize}}\quad \sum_{k = 1}^{K} \alpha_k\left(R_{k}^p + R_{k}^c\right)  \label{eq:Obj} \\
		&\text{subject to}\quad \eqref{eq:S}, \eqref{eq:S11}\\
		&\Gamma_{k}^p \geq 2^{R_{k}^p/B} - 1 \quad \forall k \in \mathcal{K} \label{eq:O1}\\
		&\Gamma_{i,k}^c \geq 2^{R_{k}^c/B} - 1 \quad \forall i \in \mathcal{M}_k \hspace{1.5mm} \text{and} \hspace{1.5mm} \forall k \in \mathcal{K} \label{eq:O2}
		\end{align}
	\end{subequations}
	where the coefficient $ \alpha_k$ refers to the priority weight associated with user $k$.
	Problem \eqref{eq:Opt1} is a mixed integer non linear problem, which is generally an NP-hard problem, due to its mixed discrete and continuous optimization nature, and the non-convexity of the underlying objective and constraints as a function of the beamforming vectors.
	%The best we can hope for, is to obtain a stationary sub-optimal solution with a good quality.
	To tackle this challenging problem, we propose an iterative algorithm based on a strongly inner-convex approximation framework coupled with a smooth approximation of the non-smooth, non-convex $l_0$-norm. Before we proceed to the technical details of our approach, we elaborate on the structure of problem \eqref{eq:Opt1}. The problem is non-convex even if we relax the binary constraints in \eqref{eq:S11}, e.g., by using $l_1$ relaxation to the $l_0$-norm. This is due to non-convexity of the objective \eqref{eq:Obj} as a function of the beamforming vectors. Moreover, the achievability constraints and the backhaul constraints in \eqref{eq:O1}-\eqref{eq:O2} and \eqref{eq:S11} are non-convex functions, and define a non-convex feasible set. To overcome this difficulty, we approximate each non-convex function with a surrogate upper-bound convex function, which helps approximating the non-convex feasible set with a  convex one. Then, we iteratively refine this approximation till convergence. The following section describes all the technicalities of the above steps in details.
	\section{Proposed Solution}
	{In this section, we present our proposed framework to tackle problem \eqref{eq:Opt1}. We start by relaxing the discrete variables, and then we proceed by introducing an inner convex approximation (ICA) reformulation of the non-convex clustering problem. After determining the clusters, we determine the optimal beamforming vectors and the RS mode to transmit private and common messages, respectively, which also quantifies how much rate is assigned to the private and common messages, respectively.}
	\subsection{Relaxing the $l_0$-norm}
	We use a smooth concave function to approximate the non-smooth, non-convex (in fact integer) $l_0$-norm. Consider the function $f_{\theta}\left(x\right)$ defined as:
	\begin{equation}
	f_{\theta}\left(x\right) = \frac{2}{\pi}\arctan\left(\frac{x}{\theta}\right), \quad x\geq 0
	\end{equation}
	which is often used in the literature to approximate the $l_0$-norm \cite{7488289}, \cite{8254789}.
	Here, $\theta$ is a smoothness parameter which controls the quality of the $l_0$-norm approximation. %If $\theta$  is large enough compared to $x$, the function $f_{\theta}\left(x\right)$ becomes smoother; however, the approximation may not be tight enough. On the other hand, the approximation precision improves if $\theta$ becomes relatively smaller than $x$. But the function would then become less smooth.
	%\subsection{Problem Reformulation}
	After relaxing the discrete $l_0$-norm, {we reformulate the problem \eqref{eq:Opt1} by introducing} SINR variables instead of using the rate expressions. Let $R_{k}^p = B\log_2\left( 1+\gamma_{k}^p\right)$ and $R_{k}^c = B\log_2\left( 1+\gamma_{k}^c\right)$ { for some $\gamma_{k}^p, \gamma_{k}^c > 0$}. Now, we can rewrite \eqref{eq:Opt1} as:
	\begin{subequations}\label{eq:Opt2}
		\begin{align}
		&\underset{\left\lbrace \mathbf{w}_{k}^p, \mathbf{w}_{k}^c, \bm{\gamma}_{k}| \forall k \in \mathcal{K}\right\rbrace}{\text{maximize}}\quad \sum_{k = 1}^{K} \alpha_k B \left(\log_2\left( 1+\gamma_{k}^p\right) + \log_2\left( 1+\gamma_{k}^c\right)\right) \nonumber \\
		&\text{subject to}\quad \eqref{eq:S}\\
		&\Gamma_{k}^p \geq \gamma_{k}^p \quad \forall k \in \mathcal{K} \label{eq:O3}\\
		&\Gamma_{i,k}^c \geq \gamma_{k}^c \quad \forall i \in \mathcal{M}_k \hspace{1.5mm} \text{and} \hspace{1.5mm} \forall k \in \mathcal{K} \label{eq:O4}\\
		&\sum\limits_{k \in \mathcal{K}} B \Bigl(f_{\theta}\left(\normV[\big]{\mathbf{w}_{n, k}^p}_{2}^{2}\right)\log_2\left( 1+\gamma_{k}^p\right)  +  \nonumber \\&
		f_{\theta}\left(\normV[\big]{\mathbf{w}_{n, k}^c}_{2}^{2}\right)\log_2\left( 1+\gamma_{k}^c\right)\Bigr) \leq C_n \quad \forall n\in \mathcal{N} \label{eq:O5}.
		\end{align}
	\end{subequations}
	%where $\bm{\gamma}_{k} = \left[\gamma_{k}^p, \gamma_{k}^c\right]^T$ is a vector grouping the SINR, when receiving the private information and common information respectively
	Problem \eqref{eq:Opt2} is still non-convex despite relaxing the binary constraints. This is because the feasible set defined by constraints
	\eqref{eq:O3}-\eqref{eq:O5} is a non-convex set. To overcome this challenge, we use some algebraic manipulations to rewrite the problem \eqref{eq:Opt2} in a form that is easier to tackle, as described next in the text.
	\subsection{{ Clustering}}
	%Splitting the SINR Constraints
	Given the SINR expressions in \eqref{eq:S3} and \eqref{eq:S6}, we can equivalently write the constraints \eqref{eq:O3} and \eqref{eq:O4} as:
	\begin{align}
	\label{eq:16}
	&\sum\limits_{j \in \mathcal{K}\setminus k}\left|\mathbf{h}_{k}^H\mathbf{w}_{j}^p \right|^2 + \sum\limits_{l \in \mathcal{K} \setminus \Phi_{k}}\left|\mathbf{h}_{k}^H\mathbf{w}_{l}^c \right|^2 + \sigma^2  - \frac{\left|\mathbf{h}_{k}^H\mathbf{w}_{k}^p \right|^2}{\gamma_{k}^p} \leq 0 \\
	\label{eq:17}
	&T_k + \sum\limits_{l \in \mathcal{K} \setminus \Phi_{k}}\left|\mathbf{h}_{k}^H\mathbf{w}_{l}^c \right|^2 + \sum\limits_{\substack{m  \in \Phi_{k}\\ \pi_{k}(m)> \pi_{k}(i)}}\left|\mathbf{h}_{k}^H\mathbf{w}_{m}^c \right|^2  %\nonumber \\& 
	-\frac{\left|\mathbf{h}_{i}^H\mathbf{w}_{k}^c \right|^2}{\gamma_{k}^c} \leq 0
	\end{align}
	Note that the function $\frac{\left|\mathbf{h}_{k}^H\mathbf{w}_{k}^p \right|^2}{\gamma_{k}^p}$ in \eqref{eq:16} is of the form $\frac{\normV[]{x}_{2}^{2}}{\beta}$, which is a convex quadratic function \cite{8032503, Book2, 8254789}. {This reformulation  is useful, because it converts the constraints \eqref{eq:O3} and \eqref{eq:O4} to a difference of convex functions}, which facilitates the inner-convex approximation. Let $\mathbf{t}_{k} = \big[t_{1,k}^p, t_{1,k}^c,\ldots t_{N,k}^p, t_{N,k}^c\big]^T$ and $\mathbf{d}_{k} = \left[d_{k}^p, d_{k}^c\right]$ be slack variables. With the help of such new variables $\mathbf{t}_{k}$ and $\mathbf{d}_{k}$, we can rewrite the optimization problem \eqref{eq:Opt2} by splitting the constraint in \eqref{eq:O5} into five simpler constraints as follows:
	\begin{subequations}\label{eq:Opt5}
		\begin{align}
		&\underset{\left\lbrace \mathbf{w}_{k}^p, \mathbf{w}_{k}^c, \bm{\gamma}_{k}, \mathbf{d}_{k}, \mathbf{t}_{k}| \forall k \in \mathcal{K}\right\rbrace}{\text{maximize}}\quad \sum_{k = 1}^{K} g_1\left(\gamma_{k}^p, \gamma_{k}^c\right)\\
		&\text{subject to}\quad \eqref{eq:S}, \eqref{eq:16}-\eqref{eq:17} \label{eq:31}\\
		& \sum\limits_{k \in \mathcal{K}} \left( t_{n,k}^p d_{k}^p + t_{n,k}^c d_{k}^c\right) \leq C_n/B \quad \forall n\in \mathcal{N} \label{eq:n1}\\
		&f_{\theta}\left(\normV[\big]{\mathbf{w}_{n, k}^p}_{2}^{2}\right) \leq t_{n,k}^p \hspace{1mm} \text{and} \hspace{1mm} f_{\theta}\left( \norm{\mathbf{w}_{n, k}^c}_{2}^{2}\right) \leq t_{n,k}^c  \label{eq:n4}\\
		& \log_2\left( 1+\gamma_{k}^p\right) \leq  d_{k}^p \label{eq:n2} \\
		&\log_2\left( 1+\gamma_{k}^c\right) \leq  d_{k}^c\quad \forall n\in \mathcal{N} \hspace{1.5mm} \text{and} \hspace{1.5mm} \forall k \in \mathcal{K}  \label{eq:n3}
		\end{align}
	\end{subequations}
	where the function $g_1\left(\gamma_{k}^p, \gamma_{k}^c\right)$ is defined as: $g_1\left(\gamma_{k}^p, \gamma_{k}^c\right) =  \alpha_k B \left(\log_2\left( 1+\gamma_{k}^p\right) + \log_2\left( 1+\gamma_{k}^c\right)\right)$.
	The following proposition illustrates how problems \eqref{eq:Opt2} and \eqref{eq:Opt5} are indeed equivalent to each other. Let $\mathbf{t}, \mathbf{d}$ be slack variables defined as: $\mathbf{t} \triangleq \left[\mathbf{t}_1^T,\ldots,\mathbf{t}_K^T \right]^T$ and $\mathbf{d} \triangleq \left[\mathbf{d}_1^T,\ldots,\mathbf{d}_K^T\right]^T$
	\begin{proposition}
		$\left(\mathbf{w}^*, \boldsymbol{\gamma}^*\right)$ is a stationary solution of \eqref{eq:Opt2} if and only if there exist $\left(\mathbf{t}^*, \mathbf{d}^*\right)$ such that  $\left(\mathbf{w}^*, \boldsymbol{\gamma}^*, \mathbf{t}^*, \mathbf{d}^*\right)$ is a stationary solution of \eqref{eq:Opt5}.
	\end{proposition}
	\begin{proof}
		The respective formulations of problems \eqref{eq:Opt2} and \eqref{eq:Opt5} share the same objective function. Moreover, the maximum transmit power constraint \eqref{eq:S} is the same in both problems. Constraints in \eqref{eq:16}-\eqref{eq:17} are equivalent mathematical manipulations of constraints \eqref{eq:O3}-\eqref{eq:O4}. Furthermore, constraint \eqref{eq:O5} is equivalent to constraints \eqref{eq:n1}--\eqref{eq:n3}, after introducing the slack variables $\mathbf{t}, \mathbf{d}$. {Therefore, optimization problems \eqref{eq:Opt2} and \eqref{eq:Opt5} are equivalent to each other.}
	\end{proof}
	{ Solving problem \eqref{eq:Opt5} helps finding the clusters which serve the private and common messages respectively for each user.}
	{But since \eqref{eq:Opt5} is a non-convex problem, we propose using ICA, so as to approximate the non-convex feasible set of problem \eqref{eq:Opt5} as described next.}
	\subsection{Inner Convex Approximations (ICA)}
	Although problem \eqref{eq:Opt5} is a non-convex problem, this paper adopts well-chosen ICA techniques to convexify its feasibility set, which is defined by constraints in \eqref{eq:31}--\eqref{eq:n3}. We start with some algebraic transformations to constraint \eqref{eq:n1}. We note that the bilinear function $ t_{n,k}^p d_{k}^p + t_{n,k}^c d_{k}^c$ can be equivalently written as
	\begin{align}\label{eq:con}
	t_{n,k}^p d_{k}^p + t_{n,k}^c d_{k}^c &= \frac{1}{2}\sum_{o \in \left\lbrace p, c \right\rbrace } \bigl[ \left(t_{n,k}^o +  d_{k}^o\right)^2 - \left(t_{n,k}^o\right)^2  - \left( d_{k}^o\right)^2 \big]
	\end{align}
	This form is equivalent to a convex plus concave functions (difference of two convex functions). We proceed by introducing a convex upper bound to the bilinear function in \eqref{eq:con}, by keeping the convex part and replacing the concave function with its first-order approximation.\\
	Let $\tilde{g}_{2n}(\mathbf{t}, \mathbf{d}, \mathbf{\tilde{t}}, \mathbf{\tilde{d}})$ be defined as:
	\begin{align}\label{eq:g2}
	\tilde{g}_{2n}(\mathbf{t}, \mathbf{d}, \mathbf{\tilde{t}}, \mathbf{\tilde{d}}) &\triangleq \sum_{k \in \mathcal{K} }\sum_{o \in \left\lbrace p, c \right\rbrace }\Bigl(\frac{1}{2}\left(t_{n,k}^o +  d_{k}^o\right)^2 - \frac{1}{2}\big(\tilde{t}_{n,k}^o\big)^2 - \frac{1}{2}\big(\tilde{d}_{k}^o\big)^2 \nonumber\\ &  - \tilde{t}_{n,k}^o \left(t_{n,k}^o - \tilde{t}_{n,k}^o\right) -\tilde{d}_{k}^o\big({d}_{k}^o-  \tilde{d}_{k}^o\big)\Bigr) - C_n/B \quad \forall n \in \mathcal{N}
	\end{align}	
	{where $(\mathbf{\tilde{t}}, \mathbf{\tilde{d}})$ are feasible fixed values, which satisfy constraints \eqref{eq:n1}--\eqref{eq:n3}.}
	\begin{proposition}
		For any feasible vectors $(\mathbf{\tilde{t}}, \mathbf{\tilde{d}})$, the function $\tilde{g}_{2n}(\mathbf{t}, \mathbf{d}, \mathbf{\tilde{t}}, \mathbf{\tilde{d}})$ satisfies:
		\begin{equation}
		\tilde{g}_{2n}(\mathbf{t}, \mathbf{d}, \mathbf{\tilde{t}}, \mathbf{\tilde{d}}) \geq \underbrace{\sum\limits_{k \in \mathcal{K}} \left( t_{n,k}^p d_{k}^p + t_{n,k}^c d_{k}^c\right) - C_n/B}_{{g}_{2n}(\mathbf{t}, \mathbf{d})}
		\end{equation}
		{for all feasible values $(\mathbf{\tilde{t}}, \mathbf{\tilde{d}})$ and all $n \in \mathcal{N}$.}
	\end{proposition}
	\begin{proof}
		We note that the function
		\begin{align}\label{eq:conP}
		g_{n, k}(\mathbf{{y}}) \triangleq \frac{1}{2}\sum_{o \in \left\lbrace p, c \right\rbrace } \bigl[ \underbrace{\left(t_{n,k}^o +  d_{k}^o\right)^2}_{g_{n, k, o}^{+}(\mathbf{y})} - \big( \underbrace{\left(t_{n,k}^o\right)^2 +\left( d_{k}^o\right)^2}_{g_{n, k, o}^{-}(\mathbf{y})}\big)  \big]
		\end{align}
		has a structure of difference of two convex functions, where both functions $g_{n, k, o}^{+}\left(\mathbf{y} \right) $ and $g_{n, k, o}^{-}\left(\mathbf{y} \right) $ are convex, and $\mathbf{y} = \left[ \mathbf{{t}}^T, \mathbf{{d}}^T\right]^T$. By keeping the convex part $g_{n, k, o}^{+}\left(\cdot \right)$ unchanged and linearising the concave part $-g_{n, k, o}^{-}\left(\cdot \right)$ using the first order approximation around the point $\big( \tilde{t}_{n,k}^o, \tilde{d}_{k}^o\big) \hspace{1mm} \forall o \in \left\lbrace p, c \right\rbrace $, we get the following convex upper approximation of  the function $g_{n, k}(\mathbf{{y}})$:
		\begin{align}
		\tilde{g}_{n, k}(\mathbf{{y}},\mathbf{\tilde{y}}) \triangleq & \frac{1}{2}\sum_{o \in \left\lbrace p, c \right\rbrace } g_{n, k, o}^{+}(\mathbf{{y}}) - g_{n, k, o}^{-}(\mathbf{\tilde{y}})  %\nonumber\\&
		-\nabla_{\mathbf{{y}}}g_{n, k, o}^{-}(\mathbf{\tilde{y}})^T\left(\mathbf{{y}} - \mathbf{\tilde{y}} \right),
		\end{align}
		where  $\mathbf{\tilde{y}} = \big[ \mathbf{\tilde{t}}^T,  \mathbf{\tilde{d}}^T\big]^T$. \\
		We can write the function $\tilde{g}_{n, k}(\mathbf{{y}},\mathbf{\tilde{y}})$ as
		\begin{align}
		\tilde{g}_{n, k}(\mathbf{{y}},\mathbf{\tilde{y}}) &\eqdef \sum_{o \in \left\lbrace p, c \right\rbrace }\Bigl(\frac{1}{2}\left(t_{n,k}^o +  d_{k}^o\right)^2 - \frac{1}{2}\big(\tilde{t}_{n,k}^o\big)^2 - \frac{1}{2}\big(\tilde{d}_{k}^o\big)^2 \nonumber\\ &
		-\tilde{t}_{n,k}^o \left(t_{n,k}^o - \tilde{t}_{n,k}^o\right) -\tilde{d}_{k}^o\big({d}_{k}^o- \tilde{d}_{k}^o\big)\Bigr)
		\end{align}
		Based on the convexity of $g_{n, k}^{-}(\mathbf{y})$, the following inequality follows:
		$t_{n,k}^p d_{k}^p + t_{n,k}^c d_{k}^c = g_{n, k}(\mathbf{{y}}) \leq \tilde{g}_{n, k}(\mathbf{{y}},\mathbf{\tilde{y}})$.
		This completes the proof of proposition 2.
	\end{proof}
	Afterwards, we perform distinct ICA operations for the remaining constraints. More precisely, for constraint \eqref{eq:n4}, we linearize the concave functions $f_{\theta}\left(\normV[\big]{\mathbf{w}_{n, k}^p}_{2}^{2}\right)$ and $f_{\theta}\left(\normV[\big]{\mathbf{w}_{n, k}^c}_{2}^{2}\right)$ around $\tilde{\mathbf{w}}_{n, k}^p$ and $\tilde{\mathbf{w}}_{n, k}^c$, respectively. { This leads to} the following inner-convex approximation of the set defined by the constraints in \eqref{eq:n4}:
	\begin{align}
	\tilde{g}_3\left(\mathbf{w}_{n, k}^p, \tilde{\mathbf{w}}_{n, k}^p\right) &\triangleq f_{\theta}\left( \normV[\big]{\mathbf{w}_{n, k}^p}_{2}^{2}\right) 
	+ \nabla f_{\theta}\left( \normV[\big]{\mathbf{w}_{n, k}^p}_{2}^{2}\right)\left(\normV[\big]{\mathbf{w}_{n, k}^p}_{2}^{2} -  \normV[\big]{\tilde{\mathbf{w}}_{n, k}^p}_{2}^{2}\right) \label{eq:g3}	\\
	\tilde{g}_4\left(\mathbf{w}_{n, k}^c, \tilde{\mathbf{w}}_{n, k}^c\right) &\triangleq f_{\theta}\left( \normV[\big]{\mathbf{w}_{n, k}^c}_{2}^{2}\right)
	+ \nabla f_{\theta}\left( \normV[\big]{\mathbf{w}_{n, k}^c}_{2}^{2}\right)\left(\normV[\big]{\mathbf{w}_{n, k}^c}_{2}^{2} -  \normV[\big]{\tilde{\mathbf{w}}_{n, k}^c}_{2}^{2}\right)\label{eq:g4}
	\end{align}
	We follow the same procedure with constraints \eqref{eq:n2} and \eqref{eq:n3}, where we linearize the concave functions $\log_2(1+\gamma_{k}^p)$, $\log_2(1+\gamma_{k}^c)$ around $\tilde{\gamma}_{k}^p$ and $\tilde{\gamma}_{k}^c$, respectively. We obtain the following equations which define an inner-convex approximation of the non-convex feasible set defined by constraints \eqref{eq:n2} and \eqref{eq:n3}:
	\begin{equation}\label{eq:30}
	\tilde{g}_5\left(\gamma_{k}^p, \tilde{\gamma}_{k}^p\right) \triangleq \log_2(1+\tilde{\gamma}_{k}^p) + \frac{1}{\left( 1+\tilde{\gamma}_{k}^p\right) \text{ln}(2)}\left(\gamma_{k}^p- \tilde{\gamma}_{k}^p\right) \leq 0
	\end{equation}
	\begin{equation}\label{eq:32}
	\tilde{g}_6\left(\gamma_{k}^c, \tilde{\gamma}_{k}^c\right) \triangleq \log_2(1+\tilde{\gamma}_{k}^c) + \frac{1}{\left( 1+\tilde{\gamma}_{k}^c\right) \text{ln}(2)}\left(\gamma_{k}^c- \tilde{\gamma}_{k}^c\right) \leq 0
	\end{equation}
	Concerning the SINR constraints in \eqref{eq:16} and \eqref{eq:17}, we note that if $\left(  \tilde{\mathbf{w}}, \boldsymbol{\tilde{\gamma}}\right) $ is a feasible point of \eqref{eq:Opt5}, then the following holds:
	\begin{equation}\label{eq:18}
	\frac{\left|\mathbf{h}_{k}^H\mathbf{w}_{k}^p \right|^2}{\gamma_{k}^p} \geq \frac{2\Re\left\lbrace\big( \tilde{\mathbf{w}}_{k}^p \big) ^H\mathbf{h}_{k}\mathbf{h}_{k}^H \mathbf{w}_{k}^p \right\rbrace }{\tilde{\gamma}_k^p}-\frac{\left|\mathbf{h}_{k}^H \tilde{\mathbf{w}}_{k}^p \right|^2}{\big(\tilde{\gamma}_{k}^p\big)^2} \gamma_{k}^p
	\end{equation}
	and
	\begin{equation}\label{eq:19}
	\frac{\left|\mathbf{h}_{i}^H\mathbf{w}_{k}^c \right|^2}{\gamma_{k}^c} \geq \frac{2\Re\left\lbrace\big( \tilde{\mathbf{w}}_{k}^c \big) ^H\mathbf{h}_{i}\mathbf{h}_{i}^H \mathbf{w}_{k}^c \right\rbrace }{\tilde{\gamma}_k^c}-\frac{\left|\mathbf{h}_{i}^H \tilde{\mathbf{w}}_{k}^c \right|^2}{\big(\tilde{\gamma}_{k}^c\big)^2} \gamma_{k}^c
	\end{equation}
	where $\Re\left\lbrace \cdot \right\rbrace $ is the real part of a complex number.
	Based on inequalities \eqref{eq:18} and \eqref{eq:19}, we can establish inner-convex approximations of the constraints in \eqref{eq:16} and \eqref{eq:17} as follows:
	\begin{align}
	\tilde{g}_7\left(\mathbf{w}, \gamma_{k}^p; \tilde{\mathbf{w}}, \tilde{\gamma}_{k}^p\right) &  \triangleq
	\sum\limits_{j \in \mathcal{K}\setminus k}\left|\mathbf{h}_{k}^H\mathbf{w}_{j}^p \right|^2 + \sum\limits_{l \in \mathcal{K} \setminus \Phi_{k}}\left|\mathbf{h}_{k}^H\mathbf{w}_{l}^c \right|^2 + \sigma^2 \nonumber \\ &
	- \frac{2\Re\left\lbrace\big( \tilde{\mathbf{w}}_{k}^p \big) ^H\mathbf{h}_{k}\mathbf{h}_{k}^H \mathbf{w}_{k}^p \right\rbrace }{\tilde{\gamma}_k^p}+\frac{\left|\mathbf{h}_{k}^H \tilde{\mathbf{w}}_{k}^p \right|^2}{\big(\tilde{\gamma}_{k}^p\big)^2} \gamma_{k}^p \label{eq:20}\\
	\tilde{g}_8\left(\mathbf{w}, \gamma_{k}^c;\tilde{\mathbf{w}}, \tilde{\gamma}_{k}^c\right) & \triangleq
	T_k + \sum\limits_{l \in \mathcal{K} \setminus \Phi_{k}}\left|\mathbf{h}_{k}^H\mathbf{w}_{l}^c \right|^2 \nonumber \\  &
	+  \sum\limits_{\substack{m  \in \Phi_{k}\\ \pi_{k}(m)> \pi_{k}(i)}}\left|\mathbf{h}_{k}^H\mathbf{w}_{m}^c \right|^2 + \frac{\left|\mathbf{h}_{i}^H \tilde{\mathbf{w}}_{k}^c \right|^2}{\big(\tilde{\gamma}_{k}^c\big)^2} \gamma_{k}^c \nonumber \\& -\frac{2\Re\left\lbrace\big( \tilde{\mathbf{w}}_{k}^c \big) ^H\mathbf{h}_{i}\mathbf{h}_{i}^H \mathbf{w}_{k}^c \right\rbrace }{\tilde{\gamma}_k^c} \label{eq:g8}
	\end{align}
	The next subsection presents the strongly inner-convex approximations of problem \eqref{eq:Opt2}, and describes the algorithm that solves it.
	\subsection{Strongly ICA based Algorithm}
	%The problem of associating users in the network to cluster of BSs serving the private and common messages respectively can be expressed as:
	The functions in \eqref{eq:g2}, \eqref{eq:g3}--\eqref{eq:32} and \eqref{eq:20}--\eqref{eq:g8} define a convex feasible set, which represents an inner-approximation of the non-convex feasible set of problem \eqref{eq:Opt5}. The idea of our approach is to iteratively solve the optimization problem defined with this approximation. After each iteration, we refine the ICA of the feasible set in \eqref{eq:Opt5}, and keep iterating until convergence to a stationary solution, as described next. The approximate optimization problem is defined as follows:
	\begin{subequations}\label{eq:Opt6}
		\begin{align}
		&\underset{\left\lbrace \mathbf{w}_{k}^p, \mathbf{w}_{k}^c, \bm{\gamma}_{k}, \mathbf{d}_{k}, \mathbf{t}_{k}| \forall k \in \mathcal{K}\right\rbrace}{\text{maximize}}\quad \sum_{k = 1}^{K} g_1\left(\gamma_{k}^p, \gamma_{k}^c\right) - g_9\left({\mathbf{w}}, \boldsymbol{{\gamma}}; \tilde{\mathbf{w}}, \boldsymbol{\tilde{\gamma}}\right)\\
		&\text{subject to}\quad \eqref{eq:S} \label{eq:n8}\\
		& \tilde{g}_{2n}(\mathbf{t}, \mathbf{d}, \mathbf{\tilde{t}}, \mathbf{\tilde{d}}) \leq 0 \label{eq:n5}\\
		& \tilde{g}_3\left(\mathbf{w}_{n, k}^p, \tilde{\mathbf{w}}_{n, k}^p\right) \leq 0 \label{eq:n6}\\
		& \tilde{g}_4\left(\mathbf{w}_{n, k}^c, \tilde{\mathbf{w}}_{n, k}^c\right) \leq 0 \label{eq:n7} \\
		& \tilde{g}_5\left(\gamma_{k}^p, \tilde{\gamma}_{k}^p\right)\leq 0 \\
		& \tilde{g}_6\left(\gamma_{k}^c, \tilde{\gamma}_{k}^c\right) \leq 0 \\
		& \tilde{g}_7\left(\mathbf{w}, \gamma_{k}^p; \tilde{\mathbf{w}}, \tilde{\gamma}_{k}^p\right)  \leq 0 \\
		& \tilde{g}_8\left(\mathbf{w}, \gamma_{k}^c;\tilde{\mathbf{w}}, \tilde{\gamma}_{k}^c\right) \leq 0 \label{eq:n9}
		\end{align}
	\end{subequations}
	Here, $g_9\left({\mathbf{w}}, \boldsymbol{{\gamma}}; \tilde{\mathbf{w}}, \boldsymbol{\tilde{\gamma}}\right)$ is a proximal term to assure that the objective is a strongly concave function, and is defined as follows:
	\begin{equation}\label{eq:21}
	g_9\left({\mathbf{w}}, \boldsymbol{{\gamma}}; \tilde{\mathbf{w}}, \boldsymbol{\tilde{\gamma}}\right) = \rho_1 \normV[\big]{\mathbf{w} - \tilde{\mathbf{w}}}_{2}^{2} + \rho_2 \normV[\big]{\boldsymbol{{\gamma}} - \boldsymbol{\tilde{\gamma}}}_{2}^{2}.
	\end{equation}
	%\subsection{Algorithm Description}
	Let $\mathbf{Z} = \left[\mathbf{w}^T, \boldsymbol{\gamma}^T,\mathbf{t}^T,\mathbf{d}^T\right]^T $ be a vector stacking all the optimization variables of the problem \eqref{eq:Opt6}. Let $\widehat{\mathbf{Z}}_v$ be the variables computed at iteration $v$ as the optimal solution of problem \eqref{eq:Opt6}, and let $\mathbf{\tilde{Z}} = \left[\mathbf{\tilde{w}}^T, \boldsymbol{\tilde{\gamma}}^T,\mathbf{\tilde{t}}^T,\mathbf{\tilde{d}}^T\right]^T $ be the point at which we compute the approximate solution of problem \eqref{eq:Opt6} at iteration $v$. Furthermore, let $\mathcal{Z}$ denote the convex feasible set of problem \eqref{eq:Opt6} defined by constraints \eqref{eq:n8}--\eqref{eq:n9}. The algorithm starts by initializing the vector $ \mathbf{\tilde{Z}}$, around which we compute the next iteration. The initialization process starts by computing feasible MRC beamformers for the users' messages when considering TIN scheme, and for both private and common messages when considering RS-CMD scheme. Based on this initialization, we compute the vector $\boldsymbol{\tilde{\gamma}}$ using equations \eqref{eq:S3} and \eqref{eq:S6}. Note that the sets $\left\lbrace\Phi_k \right\rbrace_{k=1}^K $ are computed using Algorithm 1. The initialization of vectors $\mathbf{\tilde{t}}, \mathbf{\tilde{d}}$ {is done solving \eqref{eq:n4}--\eqref{eq:n3} by replacing inequalities with equalities.}
	After solving the problem \eqref{eq:Opt6} at iteration $v$, we get the optimal values stacked in vector $\widehat{\mathbf{Z}}_v$. Using $\widehat{\mathbf{Z}}_v$, we compute the vector $\mathbf{\tilde{Z}}$ for the next iteration. The detailed steps of the iterative algorithm to solve problem \eqref{eq:Opt2} are summarized in Algorithm 2 description below.
\begin{algorithm}[h]
\caption{Inner convex approximation of \eqref{eq:Opt2}.}
\begin{algorithmic}[1]
\State Initialize:  $v\leftarrow 0$, $\mathbf{\tilde{Z}} \in \mathcal{Z}$, $\rho_1 > 0, \rho_2 > 0$, $\xi \ll 1, \xi \in \mathbb{R}_+$ and $\theta= \theta_v$.
\While{$\mathbf{\widehat{Z}}_{v}$ not a stationary solution of \eqref{eq:Opt2}}
\State Solve the convex problem \eqref{eq:Opt6} and compute $\mathbf{\widehat{Z}}_{v}$
\State $\mathbf{\tilde{Z}} \leftarrow \mathbf{\tilde{Z}} + \beta_v\left(\widehat{\mathbf{Z}}_{v} - \mathbf{\tilde{Z}} \right) $  for some $\beta_v \in (0,1]$
\If{$\theta_v \geq \xi$}
 $\theta_v = \delta \theta_v$ and $\delta \in \left(0, 1\right)  $
 \EndIf
\State $v\leftarrow v+1$
\EndWhile
\end{algorithmic}
\label{Alg2}
\end{algorithm}
{The following theorem proves that Algorithm 2 produces a stationary solution of problem \eqref{eq:Opt2}.}
{
	\begin{theorem}
		Let  $\rho_1, \rho_2 > 0$, and let the step size sequence $\left\lbrace \beta_v \right\rbrace$ satisfy $\beta_v \in (0,1]$, $\beta_v \rightarrow 0$, and $\sum\nolimits_{v}\beta_v = +\infty$. Then $\left\lbrace \widehat{\mathbf{Z}}_v\right\rbrace $, the sequence generated by Algorithm \ref{Alg2}, is bounded, and converges to $\left\lbrace \widehat{\mathbf{Z}}_{v}^*,\right\rbrace $, which is a stationary solution of problem \eqref{eq:Opt6}, such that $\left(\boldsymbol{\gamma}^*, \mathbf{t}^*, \mathbf{d}^*\right) > 0$. Therefore, (according to proposition 1), $\left(\mathbf{w}^*, \boldsymbol{\gamma}^*\right)$ is also a stationary point of problem \eqref{eq:Opt2}.
	\end{theorem}
}
\begin{proof}
	The steps of the proof rely on showing that the objective and constraints of problem \eqref{eq:Opt6} satisfy the conditions of \cite[Sec. II]{7776948}, which would guarantee the convergence to a stationary point as illustrated in \cite[Theorem 2]{7776948}. Towards this end, we show next that the function $\tilde{g}_{2n}(\mathbf{t}, \mathbf{d}; \mathbf{\tilde{t}}, \mathbf{\tilde{d}})$ satisfies the following properties:
	\begin{itemize}
		\item[$\mathrm{C}1 )$]   $\tilde{g}_{2n}(\mathbf{\tilde{y}},\mathbf{\tilde{y}}) = g_{{2n}}(\mathbf{\tilde{y}})$ 
		\item[$\mathrm{C}2 )$]  $\tilde{g}_{2n}(\mathbf{{y}},\mathbf{\tilde{y}}) \geq g_{2n}(\mathbf{{y}}),\quad  \forall \mathbf{\tilde{y}} \in \mathcal{Z}$
		\item[$\mathrm{C}3 )$]  $\tilde{g}_{2n}(\bullet,\mathbf{\tilde{y}})$ is a convex function,$\quad  \forall$  $\mathbf{\tilde{y}} \in \mathcal{Z}$
		\item[$\mathrm{C}4 )$]  $\tilde{g}_{2n}(\bullet,\bullet)$ is a continuous function on the feasible set.
		\item[$\mathrm{C}5 )$]  $\nabla_{\mathbf{{y}}}\,\tilde{g}_{2n}(\mathbf{\tilde{y}},\mathbf{\tilde{y}})= \nabla_{\mathbf{{y}}}\,{g}_{2n}(\mathbf{\tilde{y}}) $
		\item[$\mathrm{C}6 )$] The function $\nabla_{\mathbf{{y}}}\,\tilde{g}_{2n}(\bullet,\bullet)$ is continuous on the feasible set
	\end{itemize}
	C1 is verified by substituting $\mathbf{y} = \left[ \mathbf{{t}}^T, \mathbf{{d}}^T\right]^T$ in \eqref{eq:g2} by $\mathbf{\tilde{y}} = \big[ \mathbf{\tilde{t}}^T,  \mathbf{\tilde{d}}^T\big]^T$. Comparing the result with $g_{2n}(\mathbf{\tilde{y}})$ then yields the equality. C2 follows directly from proposition 2. C3 also holds, since the function $\tilde{g}_{2n}(\bullet,\mathbf{\tilde{y}})$ with fixed $\mathbf{\tilde{y}}$ consists of a convex quadratic function plus a linear function, which is convex. Further, the function $\tilde{g}_{2n}(\bullet,\bullet)$ is a difference of two convex functions, and so C4 is also true. Finally, to prove C5 and C6, take the partial derivative of the function $g_{2n}(\bullet,\mathbf{\tilde{y}})$ as follows:
	{
		\begin{align} \label{eq:g5}
		\nabla_{\mathbf{{y}}}\,\tilde{g}_{2n}(\mathbf{{y}},\mathbf{\tilde{y}})= \begin{cases}
		\frac{\partial \tilde{g}_{2n}(\bullet,\mathbf{\tilde{y}})}{\partial {t_{n,k}}} &\triangleq \sum\limits_{k \in \mathcal{K} }\sum\limits_{o \in \left\lbrace p, c \right\rbrace }\bigl(\bigl(t_{n,k}^o +  d_{k}^o\bigr)
		- \tilde{t}_{n,k}^o\bigr) \quad \forall n\\ \frac{\partial   \tilde{g}_{2n}(\bullet,\mathbf{\tilde{y}})}{\partial {d_k}} &\triangleq \sum\limits_{k \in \mathcal{K} }\sum\limits_{o \in \left\lbrace p, c \right\rbrace }\bigl(\bigl(t_{n,k}^o +  d_{k}^o\bigr) 
		- \tilde{d}_{k}^o\bigr) 
		\end{cases}
		\end{align}
	}
	Similarly, the partial derivative of ${g}_{2n}(\mathbf{{y}})$ is:
	{
		\begin{align}\label{eq:S60}
		\nabla_{\mathbf{{y}}}\,{g}_{2n}(\mathbf{{y}})= \begin{cases}
		\frac{\partial \tilde{g}_{2n}(\mathbf{{y}})}{\partial {t}_{n, k}} &\triangleq \sum_{k \in \mathcal{K} }\left(d_{k}^p +  d_{k}^c\right)
		% \nonumber\\ &
		\\ \frac{\partial \tilde{g}_{2n}(\mathbf{{y}})}{\partial {d_k}} &\triangleq \sum_{k \in \mathcal{K} }\left(t_{n,k}^p +  t_{n,k}^c\right)   % \nonumber\\ &
		\end{cases}
		\end{align}
	}
	C5 then follows by substituting $\mathbf{{y}}$ with $\mathbf{\tilde{y}}$ in both \eqref{eq:g5} and \eqref{eq:S60}. C6 also holds since that the function $\nabla_{\mathbf{{y}}}\,\tilde{g}_{2n}(\bullet,\bullet)$ is bilinear. {The above proof verifies that the function $\tilde{g}_{2n}(\mathbf{t}, \mathbf{d}; \mathbf{\tilde{t}}, \mathbf{\tilde{d}})$ satisfies the properties C1-C6. One can similarly check that all other functions associated with the optimization problem \eqref{eq:Opt6} also satisfy C1-C6, which completes the proof.}
	%\begin{align}\label{eq:g6}
	%\nabla_{\mathbf{{y}}}\,{g}_2(\mathbf{{y}}) = \begin{cases}
	%\frac{\partial {g}_2(\mathbf{{y}})}{\partial \mathbf{t}} &\triangleq \sum_{k \in \mathcal{K} }\Bigl(\left(d_{k}^p +  d_{k}^c\right) \Bigr) \nonumber \\
	% \frac{\partial {g}_2(\mathbf{{y}})}{\partial \mathbf{d}} &\triangleq \sum_{k \in \mathcal{K} }\Bigl(\left(t_{n,k}^o +  d_{k}^o\right)\Bigr)
	%  \end{cases}
	%\end{align}
\end{proof}
After solving problem \eqref{eq:Opt5}, we can determine the clusters for private and common messages as follows:
\begin{align}
\mathcal{K}_{n}^{p} & = \left\lbrace k|\quad  \normV[\big]{\mathbf{w}_{n, k}^p}_{2}^{2} \geq \epsilon_1 \right\rbrace, \label{eq:e1}\\
\mathcal{K}_{n}^{c} & =  \left\lbrace k|\quad \normV[\big]{\mathbf{w}_{n, k}^c}_{2}^{2} \geq \epsilon_2 \right\rbrace \label{eq:e2},
\end{align}
{\color{black} Where $\epsilon_1$ and $\epsilon_2$ are positive constants, which are set in the simulations section to -80 dBm/Hz.}

\subsection{Beamforming and RS mode Selection}
After fixing the clusters $\mathcal{K}_{n}^{p} $ and $\mathcal{K}_{n}^{c}$ as described above, the paper now focuses on determining the beamforming vectors by revisiting problem \eqref{eq:Opt6}. Note that when the clusters are fixed, the optimization variables become the group sparse beamforming vectors $\left\lbrace \mathbf{w}_{k}^p, \mathbf{w}_{k}^c, \bm{\gamma}_{k}| \forall k \in \mathcal{K}\right\rbrace$. Mathematically, the optimization problem \eqref{eq:Opt6} for fixed clusters can be written as:
\begin{subequations}\label{eq:Opt7}
	\begin{align}
	&\underset{\left\lbrace \mathbf{w}_{k}^p, \mathbf{w}_{k}^c, \bm{\gamma}_{k}| \forall k \in \mathcal{K}\right\rbrace}{\text{maximize}}\quad \sum_{k = 1}^{K} g_1\left(\gamma_{k}^p, \gamma_{k}^c\right) - g_9\left({\mathbf{w}}, \boldsymbol{{\gamma}}; \tilde{\mathbf{w}}, \boldsymbol{\tilde{\gamma}}\right)\\
	&\text{subject to}\quad g_{10}(\gamma_{k}^p, \tilde{\gamma}_{k}^p, \gamma_{k}^c, \tilde{\gamma}_{k}^c) \leq 0\\
	&\sum\limits_{k \in \mathcal{K}_{n}^{p}}{\normV[\big]{\mathbf{w}_{n, k}^p}_{2}^{2}} + \sum\limits_{k \in \mathcal{K}_{n}^{c}}{\normV[\big]{\mathbf{w}_{n, k}^c}_{2}^{2}} \leq P_{n}^{\text{Max}} \quad \forall n\in \mathcal{N}\\
	& g_7\left(\mathbf{w}, \gamma_{k}^p; \tilde{\mathbf{w}}, \tilde{\gamma}_{k}^p\right)  \leq 0 \\
	& g_8\left(\mathbf{w}, \gamma_{k}^c;\tilde{\mathbf{w}}, \tilde{\gamma}_{k}^c\right) \leq 0 \label{eq:n10}	
	\end{align}
\end{subequations}
where $g_{10}(\gamma_{k}^p, \tilde{\gamma}_{k}^p, \gamma_{k}^c, \tilde{\gamma}_{k}^c) \leq 0$ represents the backhaul constraint, and where the function $g_{10}(\cdot)$ is defined as:
\begin{align} \label{eq:g10}
g_{10}(\gamma_{k}^p, \tilde{\gamma}_{k}^p, \gamma_{k}^c, \tilde{\gamma}_{k}^c) & \triangleq
\sum\limits_{k \in \mathcal{K}_{n}^{p}} g_5\left(\gamma_{k}^p, \tilde{\gamma}_{k}^p\right)
\nonumber \\ &
+ \sum\limits_{k \in \mathcal{K}_{n}^{c}} g_6\left(\gamma_{k}^c, \tilde{\gamma}_{k}^c\right) -  C_n/B.
\end{align}
{\color{black} We note that problem \eqref{eq:Opt7} is similar to problem \eqref{eq:Opt6}; however, the association variables are fixed here and the goal is to find beamforming vectors with good quality. Toward this goal, we suggest using Algorithm 3 shown below to obtain a stationary solution $\left(\mathbf{w}^*, \boldsymbol{\gamma}^*\right)$ to the beamforming problem with fixed clusters.}
Here, $\mathbf{Y} = \left[\mathbf{w}^T, \boldsymbol{\gamma}^T\right]^T$, $\mathbf{\tilde{Y}} = \left[\mathbf{\tilde{w}}^T, \boldsymbol{\tilde{\gamma}}^T\right]^T $ and $\mathcal{Y}$ is the feasible set of problem \eqref{eq:Opt7}.
\begin{algorithm}[h]
\caption{Inner convex approximation of beamforming problem with fixed clusters.}
\begin{algorithmic}[1]
\State Initialize:  $v\leftarrow 0$, $\mathbf{\tilde{Y}} \in \mathcal{Y}$ and $\rho_1 > 0, \rho_2 > 0$
\While{$\mathbf{\widehat{Y}}_{v}$ not a stationary solution.}
\State Solve the convex problem \eqref{eq:Opt7} and compute $\mathbf{\widehat{Y}}_{v}$
\State $\mathbf{\tilde{Y}} \leftarrow \mathbf{\tilde{Y}} + \beta_v\left(\widehat{\mathbf{Y}}_{v} - \mathbf{\tilde{Y}} \right) $  for some $\beta_v \in (0,1]$
\State $v\leftarrow v+1$
\EndWhile
\end{algorithmic}
\label{Alg3}
\end{algorithm}
\subsection{Complexity Analysis}
{\color{black} The overall approach of joint clustering, RS mode and beamforming vectors design is split into two stages. In the first one we use Algorithm 2 to find the clusters of BSs which serve the private message and common message of each user respectively. After that, we use the Algorithm 3 to find a high-quality solution of beamforming vectors and RS which are also feasible to the original problem \eqref{eq:Opt2}. In the following, we describe the overall complexity of such an approach.}

At each iteration of Algorithm 2, which is used to determine the clusters, we need to solve a convex problem, \text{precisely problem \eqref{eq:Opt6}}, which has a logarithm plus a proximal term as an objective function. The logarithmic part can be linearised as in equation \eqref{eq:30}, which gives a quadratic convex problem which can be easily cast as a second order cone program (SOCP); see \cite{8254700} and references therein. SOCP problems can be solved using interior-point methods with a complexity of $\mathcal{O}(NKL)^{3.5 }$ via general-purpose solvers, e.g. SDPT3 or MOSEK. {After clustering, the beamforming vectors and the RS mode are determined using Algorithm \ref{Alg3}, which can similarly cast as an SOCP using a similar argument as above. Let $\text{V}_{\text{max}}$ be the worst-case fixed number of iterations needed for the Algorithm 2 (or Algorithm \ref{Alg3}) to converge. The overall computational complexity to implement Algorithm \ref{Alg2} and Algorithm \ref{Alg3} becomes, therefore, $2 \text{V}_{\text{max}}(NKL)^{3.5}$. {Note this is a rather an upper bound on the complexity metric, since solving the sparse optimization problem \eqref{eq:Opt7} is typically much faster than solving problem \eqref{eq:Opt6}}, and so it needs a smaller number of iterations for convergence.}
\section{Numerical Results}
In this section, we present an extensive set of numerical simulations to demonstrate the performance of our proposed approach. The system setup considers a C-RAN consisting of a 7-cell wrapped-around network. In each cell, there exists a BS at the center, which is connected to the cloud via a limited capacity backhaul link. The simulations results illustrated in this section assume the parameters summarized in table II, unless mentioned otherwise in the text. In particular, for illustration, all BS's share the same backhaul constraint, and all BS's operate at the same nominal maximum transmission power. 

In addition to the dynamic clustering algorithms applied for both TIN and RS-CMD, in which we jointly optimize the BSs clusters together with the beamforming vectors, we also consider a static clustering algorithm. Such static clustering, considered herein as a clustering baseline approach, adopts a path-loss information-based approach, and so the beamforming vectors are optimized for fixed (static) clusters. In the following, we explain briefly the static TIN clustering as used in \cite{6920005}, and our extended version of this algorithm to fit the RS-CMD framework.
\begin{itemize}
	\item \textbf{Static TIN}: This scheme is based on clustering procedure described in \cite[Algorithm 3]{6920005}. Once the clusters are fixed, we can solve problem \eqref{eq:Opt7} to determine the optimal beamforming vectors. 
	\item \textbf{Static RS-CMD}: In this case, we extend the previous procedure to accommodate clusters for private and common messages for each user. Again, when the clusters are fixed, we use Algorithm 3 to solve problem \eqref{eq:Opt7} over the private and common beamforming vectors.
\end{itemize}
{\begin{center}
		
		\begin{tabular}{ |p{4.5cm}|p{3.2cm}|  }
			\hline
			\multicolumn{2}{|c|}{Simulation Parameters} \\
			\hline
			Network Parameter & Value\\
			\hline
			\hline
			Channel Bandwidth   & 10 MHz   \\
			\hline
			Number of Antennas &   8 \\
			\hline
			Maximum transmission Power & 30 dBm  \\
			\hline
			Antenna gain    & 15 dBi\\
			\hline
			Background noise &   -169 dBm/Hz\\
			\hline
			Path-loss & $140.7 + 36.7 \log_{10}(d)$ \\
			\hline
			Log-normal shadowing& 8 dB  \\
			\hline
			Rayleigh small scale fading& 0 dB  \\
			\hline
		\end{tabular}
	\end{center}
}
Further, we assume that each user can decode only one additional common message besides its own common message (i.e., D = 1); however, a common message of a user can be decoded by multiple users. Such strategy helps reducing the complexity of the overall algorithm, by limiting the number of successive cancellation stages.
\subsection{Impact of Backhaul Capacity}
First, we evaluate the performance of RS-CMD scheme in C-RAN, against the state-of-the art TIN scheme. For both schemes, we consider dynamic and static clustering procedures. In case of dynamic clustering, we use Algorithm 2 and equations \eqref{eq:e1} and \eqref{eq:e2} to determine the clusters. Then, we use Algorithm 3 with few iterations, starting from the solution computed at last iteration of Algorithm 2, to compute the beamforming vectors. In the case of fixed static clusters, we apply Algorithm 3 directly to compute the beamforming vectors for both TIN and RS-CMD.

We first consider a network in which the inter-cell distance between two neighboring BSs is 200m. Fig. \ref{Fig:1} shows the achievable sum rate as a function of backhaul capacity when applying these schemes, where we use RS-CMD in both static and dynamic clustering with parameter values $\mu= 25$ and $\mu= 60$. The figure shows that our proposed algorithm, namely RS-CMD with dynamic clustering and $\mu= 25$, outperforms the state-of-the art TIN. In fact, compared with static TIN, RS-CMD with $\mu= 25$ has a significant gain up to 42.3 \% at 950 Mbps backhaul capacity.

Fig. \ref{Fig:1}, particularly, distinguishes between two backhaul capacity regions. In the low backhaul capacity region, the performance is mainly limited by capacity of backhaul links. In this region, due to the scarcity of backhaul resources, a carefully chosen set of users should be assigned to each BS in order to optimize the performance with the available backhaul resources. This explains why the static clustering schemes perform poorly in this region, while the dynamic schemes achieve sum-rates which are close to the capacity upper bound, i.e., $7C_n$. As we move towards higher backhaul capacities, we note that the sum-rate of all schemes increases, especially the proposed approaches which show a significantly higher sum-rate as compared to the state-of-the art schemes under both dynamic and static clustering, i.e., static TIN and dynamic TIN. We further note that as we approach the interference limited region by increasing the backhaul capacity, the effect of RS-CMD becomes more pronounced. This is expected, since our scheme is especially designed to mitigate interference, and so its performance gets better as the interference becomes the limiting factor. Interestingly, the role of dynamic clustering becomes less significant in the interference limited regime. Here, in the 200 meters inter-cell distance, it is likely that a significant number of users have strong channel gains to nearby BSs, which have enough backhaul resources in the interference limited regime. Such observation makes the impact of clustering in this region less significant as compared to RS-CMD. Thus, we observe that RS-CMD with static clustering outperforms dynamic TIN. Finally, in Fig. \ref{Fig:1}, we observe that RS-CMD scheme with $\mu= 25$ (referred to as RS-CMD 25) performs better than RS-CMD with $\mu= 60$. Such fact is also expected, since increasing $\mu$ increases the number of users which participate in decoding the common messages of their interferers, which adds more constraints to the optimization problem, and reduces the value of the optimized objective function, as clearly illustrated in Fig. \ref{Fig:1}.
To best illustrate the performance of the proposed algorithm in a larger inter-cell distance, Fig. \ref{Fig:2} plots the sum-rate across the network versus the backhaul capacity, where the inter-cell distance is set to 400 m. In this case, the cell-edge users are more susceptible to interference from BSs in neighboring cells. On the other hand, the users located near a cell center have better channel gains to BSs in their cell, and weak interference channels to other BSs in other cells. Under such relatively large inter-cell distance, the size of clusters becomes smaller, because only few BSs have strong channel gains to each user as compared to the small inter-cell distance. This explains why Fig. \ref{Fig:2} shows that all schemes perform well in the backhaul limited regime. However, as we approach the interference limited regime, the impact of clustering and RS-CMD becomes more significant. In this network, RS-CMD with $\mu= 25$ achieves a gain up to 61.11 \% compared to static TIN, which best highlights the significant gain harvested by common message decoding in C-RAN systems, as illustrated next.
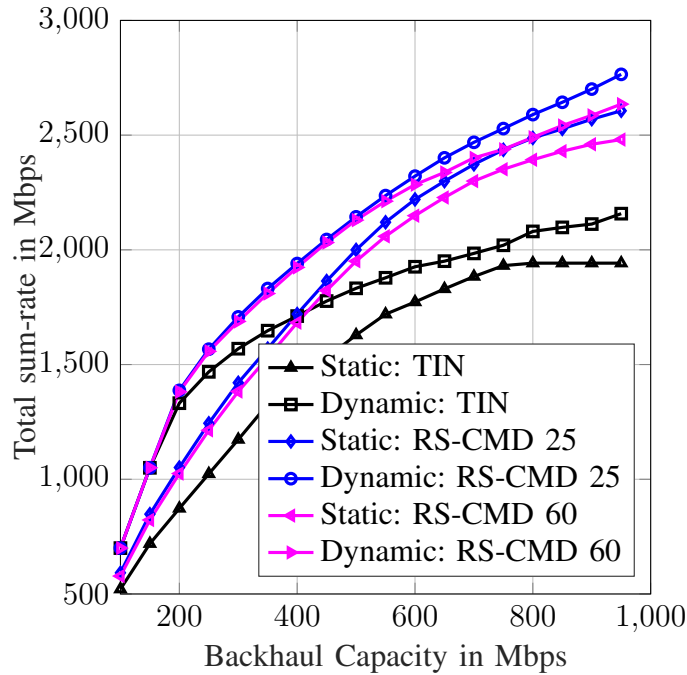
\begin{figure}
\centering
	\input{Figures/t2.tex}
	\caption{The performance of all studied schemes for a C-RAN with 7 BSs serving 28 users and an inter-cell distance of 200 m}
	\label{Fig:1}
\end{figure}
%We repeat the experiment for a different network in which the inter-cell distance is 400 m.
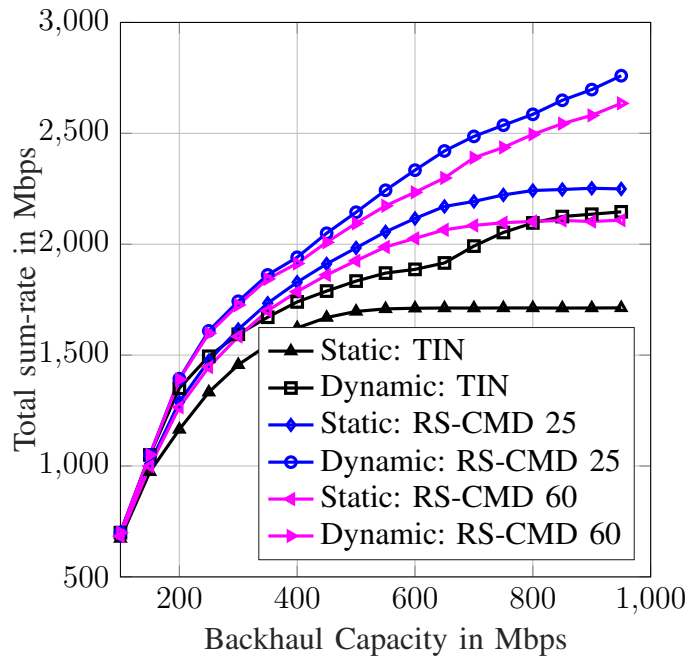
\begin{figure}
\centering
	\input{Figures/Figures/Figures/400.tex}
	\caption{The performance of all studied schemes for a C-RAN with 7 BSs serving 28 users and an inter-cell distance of 400 m}
    \label{Fig:2}
\end{figure}
\subsection{The Role of RS-CMD}
To illustrate the impact of common message decoding on the system performance, Fig. \ref{Fig:3} plots the sum rates of both the common part and the private part as a function of the backhaul capacity. The figure shows that the rate of the common message increases as the backhaul capacity increases, which highlights the impact of RS-CMD in the interference limited regime. Interestingly, as we increase the number of users which decode the common-messages of other users, i.e., as $\mu$ increases, Fig. \ref{Fig:4} shows that the rate of common messages decreases which reduces the total achievable rate, for the same reasons discussed earlier.
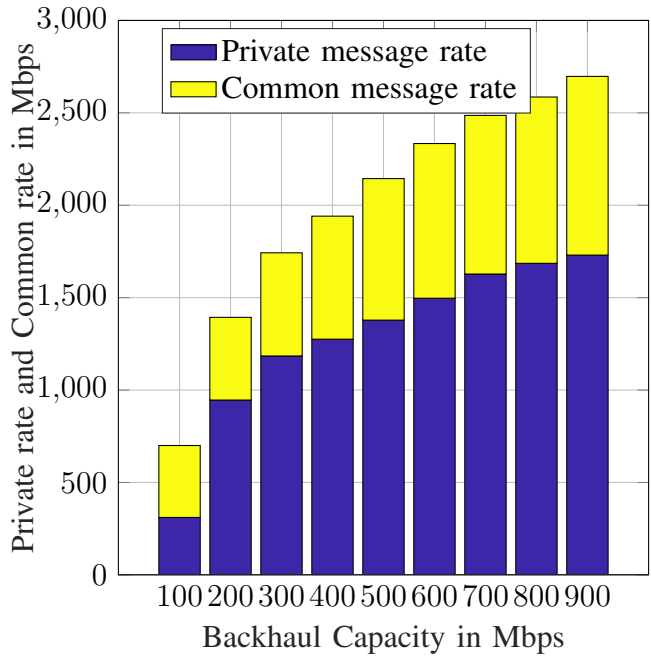
\begin{figure}[h]
\centering
	\input{Figures/Figures/Figures/400Dynamic.tex}
	\caption{The sum-rate of common message and private message using RS-CMD with $\mu= 25$ for a C-RAN with 7 BSs serving 28 users and an inter-cell distance of 400 m}
	\label{Fig:3}
\end{figure}

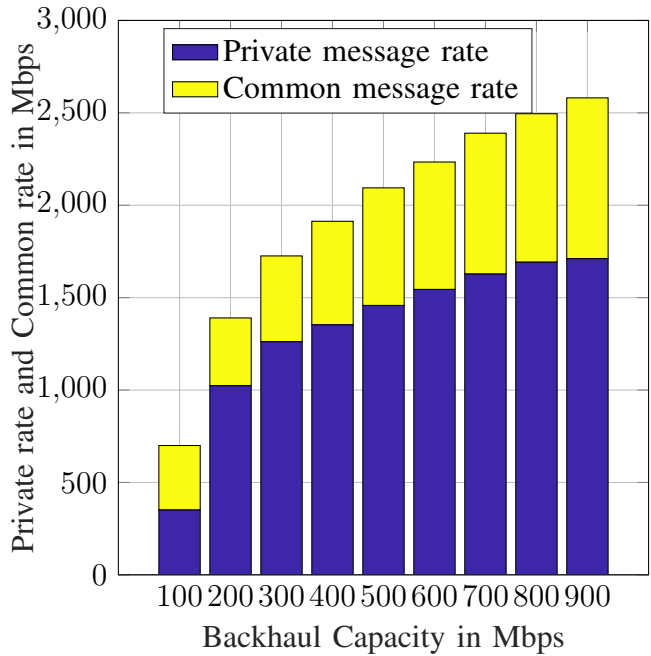
\begin{figure}[h]
\centering
	\input{Figures/Figures/Figures/400Dynamic60.tex}
	\caption{The sum-rate of common message and private message using RS-CMD with $\mu= 60$ for a C-RAN with 7 BSs serving 28 users and an inter-cell distance of 400 m}
	\label{Fig:4}	
\end{figure}
\subsection{Transmission Power Impact on RS-CMD}
Fig. \ref{Fig:5} shows the sum-rate versus the maximum transmission power, so as to study the impact of transmission power on the performance of RS-CMD. We consider a C-RAN system of 7 BSs serving 28 users. Each BS has 750 Mbps backhaul.
The inter-cell distance is set to 200 m. The figure adopts the static clustering for both TIN and RS-CMD, and shows that the gain of RS-CMD compared to TIN increases as the power increases. For $\mu= 25$, the gain of RS-CMD over TIN increases from about 12\% at 0dBm maximum transmission power, to almost 19\% at 40dBm. The rationale for such observation is that as the transmission power increases, the interference experienced in the network increases, and so the role of RS-CMD as an interference mitigation technique becomes more pronounced.
\begin{figure}
\centering
	\input{Figures/Figures/Figures/power1.tex}
	\caption{The achievable sum-rate as a function of maximum transmission power, using static TIN and RS-CMD with $\mu= 25$ for the scenario in which a C-RAN with 7 BSs serving 28 users. Each BS has 750 Mbps backhaul. The inter-cell distance is 200 m}
    \label{Fig:5}
\end{figure}
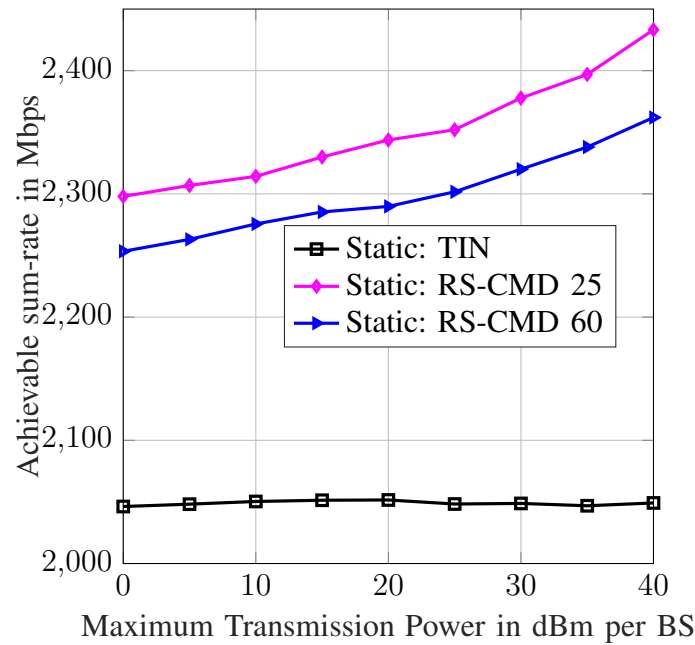
In Fig. \ref{Fig:5}, as the transmission power increases from 0 dBm to 40 dBm we see clearly that the gain of RS-CMD compared to TIN increases. In case of $\mu= 25$ the gain of RS-CMD over TIN increases from about 12\% at 0dBm maximum transmission power to almost 19\% at 40dBm. Intuitively, as the transmission power increases the interference experienced in the network increases as well. Hence, the effectiveness of RS-CMD becomes more pronounced since this method is originally designed to mitigate the interference, while using TIN becomes sub-optimal in high interference regimes.
\subsection{Convergence Behavior of Algorithm 2 and Algorithm 3}
We now illustrate the convergence behavior of Algorithms 2 and 3, as indicated in Theorem 1. We herein focus on a C-RAN system with 7 BSs serving 28 users and an inter-cell distance of 400 m. All the simulation results are averaged over 80 random realizations. In Fig. \ref{Fig:6}, we plot the objective function of problem \eqref{eq:Opt6} as a function of the number of iterations executed while implementing Algorithm \ref{Alg2}, so as to illustrate its convergence. Similarly, Fig. \ref{Fig:7} plots the objective function of problem \eqref{eq:Opt6} as a function of the number of iterations executed while implementing Algorithm \ref{Alg3}, so as to illustrate Algorithm \ref{Alg3} convergence. Both Fig. \ref{Fig:6} and Fig. \ref{Fig:7} illustrate the fast convergence of both Algorithm \ref{Alg2} and Algorithm \ref{Alg3}, respectively, which further highlight the numerical performance of our proposed algorithms.
\begin{figure}
\centering
	\input{Figures/hist10.tex}
	\caption{The objective function of \eqref{eq:Opt6}, using RS-CMD with $\mu= 25$ for the scenario in which a C-RAN with 7 BSs serving 28 users and an inter-cell distance of 400 m}
    \label{Fig:6}
\end{figure}
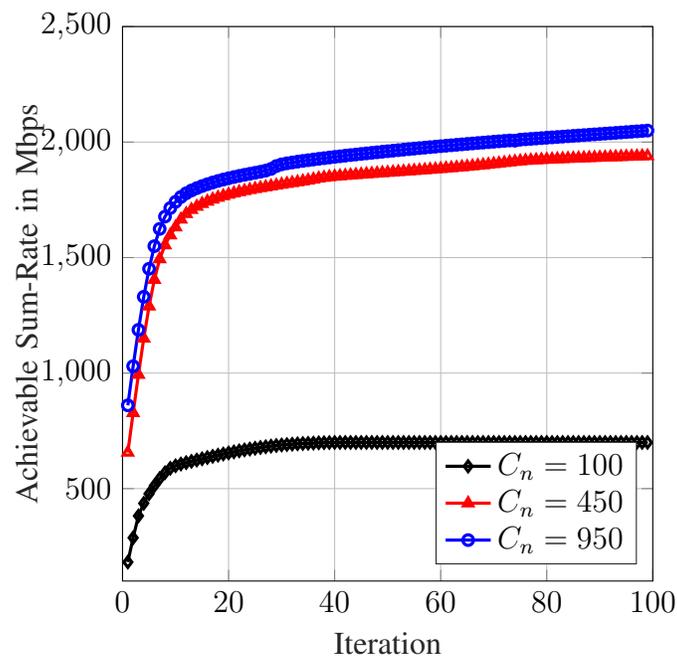
For Algorithm 3 we get the Figure 8.
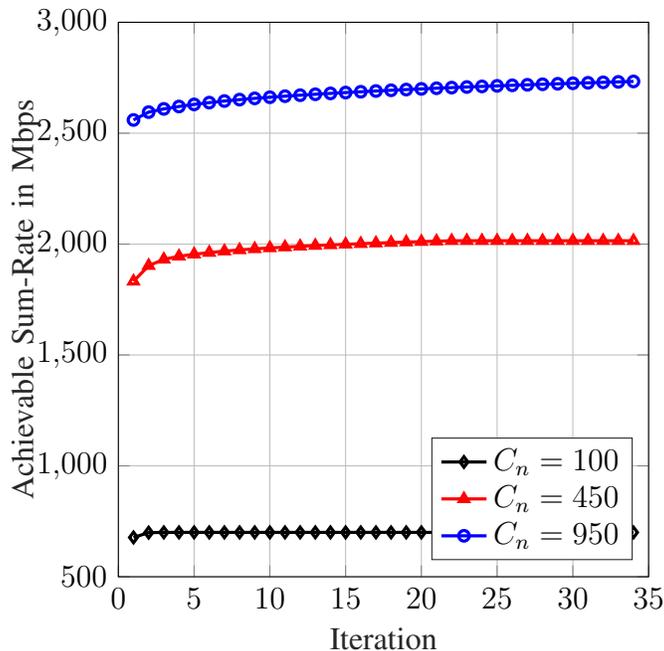
\begin{figure}
\centering
	\input{Figures/hist14.tex}
	\caption{The objective function of \eqref{eq:Opt6}, using RS-CMD with $\mu= 25$ for the scenario in which a C-RAN with 7 BSs serving 28 users and an inter-cell distance of 400 m}
    \label{Fig:7}
\end{figure}
%We clearly see from  Fig. \ref{Fig:7} that only few iterations are needed to find a stationary solution within the accuracy tolerance we choose. Interestingly, for low backhaul capacities we need much lower number of iterations to converge compared with higher backhaul capacities. The intuition behind this is that as we increase the backhaul capacity, the search space in the optimization problem becomes larger and hence it takes more effort to obtain the stationary solution.
\subsection{The impact of the Number of Users}
Last but not least, we examine the impact of increasing the number of users in the network on the achievable performance. {\color{black} In Fig. \ref{Fig:8} we clearly see that dynamic RS-CMD with $\mu= 25$ outperforms the dynamic TIN. As the number of users increases, the RS-CMD gain improves over TIN. Interestingly, as the number of users approaches the total number of transmit antennas, the gain becomes larger, i.e., when the number of users is 25 in this example. After that, when the number of users exceeds the number of transmit antennas, the achievable sum-rate by dynamic TIN saturates earlier than the dynamic RS-CMD, which further highlights the important role of joint rate splitting and common message decoding in dense networks.}
\begin{center}
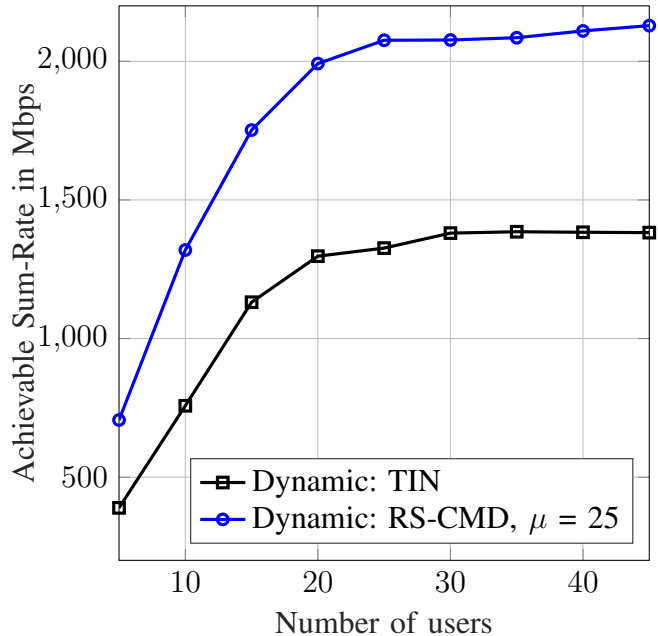
\begin{figure}
	\input{Figures/figureA12.tex}
	\caption{The achievable sum-rate, using dynamic RS-CMD with $\mu= 25$ and dynamic TIN for the scenario in which a C-RAN with 7 BSs serving 28 users. Each BS has 4 Antenna. The inter-cell distance is 200 m}
    \label{Fig:8}
\end{figure}
\end{center}
\section{Conclusions}
This paper amalgamates the benefits of RS in C-RAN for enabling large-scale interference management. We have proposed a transmission scheme for a C-RAN which capitalizes on rate-splitting, common message decoding, beamforming vector design and clustering to mitigate interference and appropriately use the limited backhaul and transmit power resources. For the proposed scheme, we formulated the problem of maximizing the weighted sum-rate subject to finite backhaul capacity and transmit power constraints. We have proposed a solution using $l_0$ relaxation followed by an ICA framework. Simulations show that the RS scheme outperforms the conventional private-information transmission approach. The gain is more significant in dense networks as well as in interference limited regimes. {\color{black}Besides, we show the benefits of joint clustering and RS mode design in enabling a better use of backhaul resources in C-RAN. This suggest RS-CMD techniques can improve the performance significantly in large and dense wireless networks.}
\bibliographystyle{IEEEtran}
\bibliography{bibliography}
\balance
\end{document}

%% file: Figures/t2.tex
%\documentclass[
%fontsize=11pt,
%paper=a4,
%pagesize=auto,
%parskip=half,
%titlepage=on,
%ngerman
%]{scrartcl}
%
%\usepackage[T1]{fontenc}
%\usepackage[utf8]{inputenc}
%\usepackage{graphicx}
%\usepackage{babel}[german]
%\usepackage{amsmath}
%\usepackage{siunitx}
%\sisetup{locale=DE}
%\usepackage{lmodern}
%\usepackage{url}
%\usepackage{microtype}
%
%% Tikz
%\usepackage{tikz}
%\usepackage{pgfplots}
%
%\begin{document}
\definecolor{mycolor1}{rgb}{1.00000,0.00000,1.00000}%
\begin{tikzpicture}

\begin{axis}[%
width=3.398in,
height=3.627in,
at={(0.471in,0.516in)},
scale = 1 , %only axis
xmin=100,
xmax=1000,
xlabel style={font=\color{white!15!black}},
xlabel={Backhaul Capacity in Mbps},
ymin=500,
ymax=3000,
ylabel style={font=\color{white!15!black}},
ylabel={Total sum-rate in Mbps},
axis background/.style={fill=white},
xmajorgrids,
ymajorgrids,
legend style={legend cell align=left, align=left, draw=white!15!black},
legend pos= south east
]
\addplot [color=black, line width=1.2pt, mark=triangle, mark options={solid, black}]
  table[row sep=crcr]{%
100	521.130461538718\\
150	719.172213206136\\
200	872.878863117243\\
250	1023.48141326181\\
300	1173.21970262788\\
350	1326.77344241968\\
400	1427.88166064789\\
450	1529.89279567917\\
500	1628.67530778376\\
550	1719.37038820565\\
600	1772.6615280402\\
650	1829.75418772304\\
700	1883.67753568413\\
750	1931.23209770592\\
800	1942.0052743667\\
850	1942.00492147335\\
900	1942.00528380571\\
950	1942.00528853911\\
};
\addlegendentry{Static: TIN}

\addplot [color=black, line width=1.2pt, mark=square, mark options={solid, black}]
  table[row sep=crcr]{%
100	700.000000000000\\
150	1049.99955301343\\
200	1331.66420467329\\
250	1468.2610892887\\
300	1569.08905270233\\
350	1647.61275526232\\
400	1711.2656236759\\
450	1776.65604229617\\
500	1832.31854103045\\
550	1878.18713331867\\
600	1926.27545731817\\
650	1950.69428431004\\
700	1984.57530773848\\
750	2019.95631943824\\
800	2080.34055763866\\
850	2097.88215379049\\
900	2112.34471613307\\
950	2157.70145553076\\
};
\addlegendentry{Dynamic: TIN}

\addplot [color=blue, line width=1.2pt, mark=diamond, mark options={solid, blue}]
  table[row sep=crcr]{%
100	592.345809882608\\
150	847.963179426302\\
200	1051.25939611577\\
250	1244.11503890578\\
300	1420.71462490165\\
350	1569.72372207245\\
400	1720.23030014541\\
450	1864.17198095936\\
500	1999.76647602218\\
550	2120.32536420823\\
600	2219.7669037913\\
650	2299.24101056582\\
700	2372.63269894403\\
750	2435.96370369262\\
800	2487.66288663663\\
850	2527.14454043935\\
900	2569.44097273318\\
950	2606.06857047616\\
};
\addlegendentry{Static:  RS-CMD 25}

\addplot [color=blue, line width=1.2pt, mark=o, mark options={solid, blue}]
  table[row sep=crcr]{%
100	699.999999723919\\
150	1049.99999962279\\
200	1387.46361144578\\
250	1566.99407894065\\
300	1707.82421789541\\
350	1831.34445387554\\
400	1940.27894942541\\
450	2045.07241690507\\
500	2143.87535477877\\
550	2236.4276241738\\
600	2321.2040493497\\
650	2401.0130953672\\
700	2468.69295202194\\
750	2528.52506890446\\
800	2589.47086927709\\
850	2643.18988163853\\
900	2700.72354007572\\
950	2764.8652702201\\
};
\addlegendentry{Dynamic: RS-CMD 25}

\addplot [color=mycolor1, line width=1.2pt, mark=triangle, mark options={solid, rotate=90, mycolor1}]
  table[row sep=crcr]{%
100	578.146423614705\\
150	822.545873872449\\
200	1024.89669618023\\
250	1213.48133511015\\
300	1382.46751996317\\
350	1530.83320775309\\
400	1682.70200303129\\
450	1822.82595967294\\
500	1950.79786775306\\
550	2058.95226401185\\
600	2148.87870329703\\
650	2228.16254466261\\
700	2299.85793045079\\
750	2351.27638377444\\
800	2392.85360964534\\
850	2429.85355195681\\
900	2460.21393710048\\
950	2480.65576176754\\
};
\addlegendentry{Static: RS-CMD 60}

\addplot [color=mycolor1, line width=1.2pt, mark=triangle, mark options={solid, rotate=270, mycolor1}]
  table[row sep=crcr]{%
100	699.999999452296\\
150	1049.99999677896\\
200	1381.25536878325\\
250	1557.58722662779\\
300	1688.19751443136\\
350	1809.64799474146\\
400	1923.52948631252\\
450	2031.77999210639\\
500	2129.50122310027\\
550	2212.09404917809\\
600	2284.17338867144\\
650	2337.35926056367\\
700	2400.91545286976\\
750	2438.06750665592\\
800	2490.51762867892\\
850	2541.97260786391\\
900	2586.16179804772\\
950	2635.09981804307\\
};
\addlegendentry{Dynamic:  RS-CMD 60}

\end{axis}
\end{tikzpicture}%
%
%\end{document}

%% file: Figures/Figures/Figures/400.tex
\definecolor{mycolor1}{rgb}{1.00000,0.00000,1.00000}%
\begin{tikzpicture}

\begin{axis}[%
width=3.398in,
height=3.527in,
at={(0.471in,0.516in)},
scale = 1 , %only axis
xmin=100,
xmax=1000,
xlabel style={font=\color{white!15!black}},
xlabel={Backhaul Capacity in Mbps},
ymin=500,
ymax=3000,
ylabel style={font=\color{white!15!black}},
ylabel={Total sum-rate in Mbps},
axis background/.style={fill=white},
xmajorgrids,
ymajorgrids,
legend style={at={(0.47,0.247)}, anchor=south west, legend cell align=left, align=left, draw=white!15!black},
legend pos= south east
]
\addplot [color=black, line width=1.2pt, mark=triangle, mark options={solid, black}]
  table[row sep=crcr]{%
100	673.084062504045\\
150	973.781592695328\\
200	1164.85088120129\\
250	1332.75124308304\\
300	1456.25369295679\\
350	1545.77823331295\\
400	1622.51219710153\\
450	1669.32073658826\\
500	1696.94345729242\\
550	1708.39768722079\\
600	1711.31720592546\\
650	1712.34441611027\\
700	1712.06168856123\\
750	1712.60102184174\\
800	1712.57326563505\\
850	1712.08557881241\\
900	1712.44894153624\\
950	1712.98356489681\\
};
\addlegendentry{Static: TIN}

\addplot [color=black, line width=1.2pt, mark=square, mark options={solid, black}]
  table[row sep=crcr]{%
100	699.398602928899\\
150	1049.80608611136\\
200	1350.87380001752\\
250	1494.50426189469\\
300	1593.91633937845\\
350	1670.81828861986\\
400	1739.79132583867\\
450	1788.7217964419\\
500	1833.8212042161\\
550	1869.84139176853\\
600	1886.57089878533\\
650	1915.54626845489\\
700	1991.52359604211\\
750	2052.9857790631\\
800	2095.9590655054\\
850	2124.70297538752\\
900	2135.18882858291\\
950	2145.94121610838\\
};
\addlegendentry{Dynamic: TIN}

\addplot [color=blue, line width=1.2pt, mark=diamond, mark options={solid, blue}]
  table[row sep=crcr]{%
100	687.22267656296\\
150	1011.08979324158\\
200	1288.50672750262\\
250	1474.64204608777\\
300	1615.76366164816\\
350	1734.03439300216\\
400	1828.66311934323\\
450	1911.18751656314\\
500	1983.42723891793\\
550	2055.7140073212\\
600	2116.04755677997\\
650	2170.25552406835\\
700	2192.80451381729\\
750	2222.21688896694\\
800	2242.10773916356\\
850	2246.72498240683\\
900	2252.1694011709\\
950	2249.29844833492\\
};
\addlegendentry{Static: RS-CMD 25}

\addplot [color=blue, line width=1.2pt, mark=o, mark options={solid, blue}]
  table[row sep=crcr]{%
100	699.999998306641\\
150	1049.99999958542\\
200	1393.38350023744\\
250	1609.22563381344\\
300	1742.63028008506\\
350	1861.55731565466\\
400	1940.97144637327\\
450	2049.41978396397\\
500	2144.15618188253\\
550	2243.29175051538\\
600	2333.86779183455\\
650	2420.42614595974\\
700	2486.15893782016\\
750	2536.51929694976\\
800	2585.6281449655\\
850	2648.94791460557\\
900	2697.16682993817\\
950	2759.74492752895\\
};
\addlegendentry{Dynamic: RS-CMD 25}

\addplot [color=mycolor1, line width=1.2pt, mark=triangle, mark options={solid, rotate=90, mycolor1}]
  table[row sep=crcr]{%
100	683.667004377584\\
150	999.396419878925\\
200	1265.01241819581\\
250	1446.17643334548\\
300	1583.89205158909\\
350	1699.92936330818\\
400	1786.02443381299\\
450	1861.21786663395\\
500	1925.56997501014\\
550	1987.63794982423\\
600	2026.3069412731\\
650	2064.96100668824\\
700	2084.64044685176\\
750	2097.06527046255\\
800	2102.47599732789\\
850	2107.27879627586\\
900	2102.36871388621\\
950	2108.50302197113\\
};
\addlegendentry{Static: RS-CMD 60}

\addplot [color=mycolor1, line width=1.2pt, mark=triangle, mark options={solid, rotate=270, mycolor1}]
  table[row sep=crcr]{%
100	699.999998109579\\
150	1049.99999645073\\
200	1390.37266151013\\
250	1597.86299187961\\
300	1725.84025477605\\
350	1843.4871876204\\
400	1913.28115948805\\
450	2008.66397403353\\
500	2094.14571834496\\
550	2173.16708884078\\
600	2233.88820747848\\
650	2298.67705103615\\
700	2389.87872056595\\
750	2436.31526070736\\
800	2495.25623667964\\
850	2543.57926942167\\
900	2581.00201358348\\
950	2635.16869521345\\
};
\addlegendentry{Dynamic: RS-CMD 60}

\end{axis}
\end{tikzpicture}%

%% file: Figures/Figures/Figures/400Dynamic.tex
\definecolor{mycolor1}{rgb}{0.24220,0.15040,0.66030}%
\definecolor{mycolor2}{rgb}{0.97690,0.98390,0.08050}%
\begin{tikzpicture}

\begin{axis}[%
width=3.398in,
height=3.527in,
at={(0.771in,0.516in)},
scale = 1 , %only ax
bar width=0.55cm,
xmin=-20,
xmax=1020,
xtick={100, 200, 300, 400, 500, 600, 700, 800, 900},
xlabel style={font=\color{white!15!black}},
xlabel={Backhaul Capacity in Mbps},
ymin=0,
ymax=3000,
ylabel style={font=\color{white!15!black}},
ylabel={Private rate and Common rate in Mbps},
axis background/.style={fill=white},
xmajorgrids,
ymajorgrids,
legend style={at={(0.082,0.832)}, anchor=south west, legend cell align=left, align=left, draw=white!15!black}
]
\addplot[ybar stacked, fill=mycolor1, draw=black, area legend] table[row sep=crcr] {%
100	310.801116968573\\
200	945.974246013299\\
300	1184.47733492753\\
400	1275.31316026797\\
500	1378.37257742236\\
600	1496.98136810817\\
700	1627.8509010919\\
800	1685.73957315338\\
900	1730.30888404435\\
};
\addplot[forget plot, color=white!15!black] table[row sep=crcr] {%
-20	0\\
1020	0\\
};
\addlegendentry{Private message rate}

\addplot[ybar stacked, fill=mycolor2, draw=black, area legend] table[row sep=crcr] {%
100	389.198881338069\\
200	447.40925422414\\
300	558.14537865185\\
400	665.630904001988\\
500	765.692842551088\\
600	836.858869912573\\
700	858.307625271039\\
800	899.886815554301\\
900	966.857945893818\\
};
\addplot[forget plot, color=white!15!black] table[row sep=crcr] {%
-20	0\\
1020	0\\
};
\addlegendentry{Common message rate}

\end{axis}
\end{tikzpicture}%

%% file: Figures/Figures/Figures/400Dynamic60.tex
\definecolor{mycolor1}{rgb}{0.24220,0.15040,0.66030}%
\definecolor{mycolor2}{rgb}{0.97690,0.98390,0.08050}%
\begin{tikzpicture}

\begin{axis}[%
width=3.398in,
height=3.527in,
at={(0.771in,0.516in)},
scale = 1 , %only ax
bar width=0.55cm,
xmin=-20,
xmax=1020,
xtick={100, 200, 300, 400, 500, 600, 700, 800, 900},
xlabel style={font=\color{white!15!black}},
xlabel={Backhaul Capacity in Mbps},
ymin=0,
ymax=3000,
ylabel style={font=\color{white!15!black}},
ylabel={Private rate and Common rate in Mbps},
axis background/.style={fill=white},
xmajorgrids,
ymajorgrids,
legend style={at={(0.085,0.832)}, anchor=south west, legend cell align=left, align=left, draw=white!15!black}
]
\addplot[ybar stacked, fill=mycolor1, draw=black, area legend] table[row sep=crcr] {%
100	352.260447761441\\
200	1024.15420984988\\
300	1262.28028339145\\
400	1353.87927512695\\
500	1458.01346106726\\
600	1544.82701007015\\
700	1628.97071517746\\
800	1693.10651786749\\
900	1711.64953423792\\
};
\addplot[forget plot, color=white!15!black] table[row sep=crcr] {%
-20	0\\
1020	0\\
};
\addlegendentry{Private message rate}

\addplot[ybar stacked, fill=mycolor2, draw=black, area legend] table[row sep=crcr] {%
100	347.739550348139\\
200	366.218451660255\\
300	463.534547808295\\
400	559.33917977564\\
500	636.093733152171\\
600	689.061197408328\\
700	760.908005388488\\
800	802.14971881215\\
900	869.352479345565\\
};
\addplot[forget plot, color=white!15!black] table[row sep=crcr] {%
-20	0\\
1020	0\\
};
\addlegendentry{Common message rate}

\end{axis}
\end{tikzpicture}%

%% file: Figures/Figures/Figures/power1.tex
\definecolor{mycolor1}{rgb}{1.00000,0.00000,1.00000}%
\begin{tikzpicture}

\begin{axis}[%
width=3.398in,
height=3.527in,
at={(0.771in,0.516in)},
scale = 1 , %only ax
xmin=0,
xmax=40,
xlabel style={font=\color{white!15!black}},
xlabel={Maximum Transmission Power in dBm per BS},
ymin=2000,
ymax=2450,
ylabel style={font=\color{white!15!black}},
ylabel={Achievable sum-rate in Mbps},
axis background/.style={fill=white},
xmajorgrids,
ymajorgrids,
legend style={legend cell align=left, align=left, draw=white!15!black,at={(0.93,0.5)},anchor=east},
%legend pos= mid east
]
\addplot [color=black, line width=1.2pt, mark=square, mark options={solid, black}]
  table[row sep=crcr]{%
0	2046.26838497563\\
5	2048.17549746922\\
10	2050.35617070157\\
15	2051.36228065737\\
20	2051.5763670141\\
25	2048.29650218624\\
30	2048.77214248961\\
35	2046.87283950311\\
40	2049.17080168042\\
};
\addlegendentry{Static: TIN}

\addplot [color=mycolor1, line width=1.2pt, mark=diamond, mark options={solid, mycolor1}]
  table[row sep=crcr]{%
0	2297.94133677458\\
5	2306.83189499698\\
10	2314.18916268185\\
15	2329.94050167906\\
20	2343.75135567227\\
25	2352.04782800815\\
30	2377.76519804729\\
35	2396.9490873371\\
40	2433.19322948391\\
};
\addlegendentry{Static: RS-CMD 25}

\addplot [color=blue, line width=1.2pt, mark=triangle, mark options={solid, rotate=270, blue}]
  table[row sep=crcr]{%
0	2253.39283887284\\
5	2263.08191758939\\
10	2275.58087222217\\
15	2285.32880041086\\
20	2289.82712100135\\
25	2301.68398768455\\
30	2320.0381150315\\
35	2337.975336534\\
40	2362.00967510618\\
};
\addlegendentry{Static: RS-CMD 60}

\end{axis}
\end{tikzpicture}%

%% file: Figures/hist10.tex
\begin{tikzpicture}

\begin{axis}[%
width=3.398in,
height=3.527in,
at={(0.758in,0.481in)},
scale = 1,
xmin=0,
xmax=100,
xlabel style={font=\color{white!15!black}},
xlabel={Iteration},
ymin=100,
ymax=2500,
ylabel style={font=\color{white!15!black}},
ylabel={Achievable Sum-Rate in Mbps},
axis background/.style={fill=white},
xmajorgrids,
ymajorgrids,
legend style={at={(0.97,0.03)}, anchor=south east, legend cell align=left, align=left, draw=white!15!black}
]
\addplot [color=black, line width=1.2pt, mark=diamond, mark options={solid, black}]
  table[row sep=crcr]{%
1	181.966140150589\\
2	287.223801026556\\
3	381.939248176794\\
4	436.363541624644\\
5	477.553094123384\\
6	512.567427641846\\
7	543.518573007223\\
8	568.609640826638\\
9	587.977251863106\\
10	598.85646414276\\
11	606.449508090416\\
12	612.42448352025\\
13	618.040641217932\\
14	623.576088037258\\
15	629.042487199455\\
16	634.42464795352\\
17	639.694794317544\\
18	644.763656036\\
19	649.643285853633\\
20	654.366571445915\\
21	658.960437178867\\
22	663.400206314433\\
23	667.545900771075\\
24	671.298448182968\\
25	674.734406449203\\
26	678.030526056126\\
27	681.126868577751\\
28	684.005299646463\\
29	686.641493057087\\
30	688.83163222565\\
31	690.580201207926\\
32	692.229001232103\\
33	693.842224986159\\
34	695.09535654432\\
35	696.034142392364\\
36	696.930615318343\\
37	698.10227194382\\
38	698.672025750983\\
39	699.273727156705\\
40	699.73250071434\\
41	699.73250071434\\
42	699.73250071434\\
43	699.73250071434\\
44	699.73250071434\\
45	699.73250071434\\
46	699.73250071434\\
47	699.73250071434\\
48	699.73250071434\\
49	699.73250071434\\
50	699.73250071434\\
51	699.73250071434\\
52	699.73250071434\\
53	699.73250071434\\
54	699.73250071434\\
55	699.73250071434\\
56	699.73250071434\\
57	699.73250071434\\
58	699.73250071434\\
59	699.73250071434\\
60	699.73250071434\\
61	699.73250071434\\
62	699.73250071434\\
63	699.73250071434\\
64	699.73250071434\\
65	699.73250071434\\
66	699.73250071434\\
67	699.73250071434\\
68	699.73250071434\\
69	699.73250071434\\
70	699.73250071434\\
71	699.73250071434\\
72	699.73250071434\\
73	699.73250071434\\
74	699.73250071434\\
75	699.73250071434\\
76	699.73250071434\\
77	699.73250071434\\
78	699.73250071434\\
79	699.73250071434\\
80	699.73250071434\\
81	699.73250071434\\
82	699.73250071434\\
83	699.73250071434\\
84	699.73250071434\\
85	699.73250071434\\
86	699.73250071434\\
87	699.73250071434\\
88	699.73250071434\\
89	699.73250071434\\
90	699.73250071434\\
91	699.73250071434\\
92	699.73250071434\\
93	699.73250071434\\
94	699.73250071434\\
95	699.73250071434\\
96	699.73250071434\\
97	699.73250071434\\
98	699.73250071434\\
99	699.73250071434\\
};
\addlegendentry{$C_n = 100$}

\addplot [color=red,line width=1.2pt, mark=triangle, mark options={solid, red}]
  table[row sep=crcr]{%
1	656.425174214934\\
2	827.909974083853\\
3	993.365407864216\\
4	1150.29903391684\\
5	1289.20304317259\\
6	1404.41153373411\\
7	1493.68641759502\\
8	1553.95463376597\\
9	1597.05451594734\\
10	1631.76873025099\\
11	1665.46249800949\\
12	1689.3474216564\\
13	1707.1233215944\\
14	1721.00664528429\\
15	1732.7925412606\\
16	1742.98997922343\\
17	1751.85992151798\\
18	1759.70981460488\\
19	1766.72570322686\\
20	1773.03754518376\\
21	1778.75809963877\\
22	1783.98623546406\\
23	1788.80799545729\\
24	1793.29687891299\\
25	1797.51524784488\\
26	1801.51579500866\\
27	1805.34498428542\\
28	1809.05619274219\\
29	1812.68750164985\\
30	1816.2598047202\\
31	1819.79437563788\\
32	1823.31651363878\\
33	1826.86063265839\\
34	1830.48180048195\\
35	1834.28218325348\\
36	1838.46884000871\\
37	1842.98993820867\\
38	1846.25136836409\\
39	1848.9180181686\\
40	1850.93076990641\\
41	1852.8106960475\\
42	1854.62282665516\\
43	1856.38415129359\\
44	1858.10377445483\\
45	1859.78855034646\\
46	1861.44376990039\\
47	1863.07410630698\\
48	1864.68491422803\\
49	1866.28163009114\\
50	1867.86758520996\\
51	1869.44863045497\\
52	1871.02914858707\\
53	1872.61392855712\\
54	1874.20841054821\\
55	1875.81626052292\\
56	1877.44416727759\\
57	1879.09683806893\\
58	1880.78173687473\\
59	1882.50018967732\\
60	1884.26251747062\\
61	1886.07266157589\\
62	1887.93553639141\\
63	1889.8603554431\\
64	1891.85019091169\\
65	1893.91085223548\\
66	1896.04528084406\\
67	1898.25487720461\\
68	1900.54235095106\\
69	1902.90405861834\\
70	1905.33807809298\\
71	1907.81823834214\\
72	1910.33573702698\\
73	1912.8445752798\\
74	1915.26388369205\\
75	1917.52547565945\\
76	1919.36604194939\\
77	1920.86024832199\\
78	1922.19597592034\\
79	1923.33812456475\\
80	1924.41259216432\\
81	1925.43515741555\\
82	1926.41649772055\\
83	1927.3606416473\\
84	1928.27676257638\\
85	1929.16589200183\\
86	1930.03081254219\\
87	1930.87529738696\\
88	1931.70495852952\\
89	1932.51181447159\\
90	1933.3066088062\\
91	1934.08363406579\\
92	1934.84713540037\\
93	1935.59980242264\\
94	1936.34144752271\\
95	1937.06857913499\\
96	1937.78642750377\\
97	1938.49714665866\\
98	1939.19308746293\\
99	1939.8814870688\\
};
\addlegendentry{$C_n = 450$}

\addplot [color=blue,line width=1.2pt, mark=o, mark options={solid, blue}]
  table[row sep=crcr]{%
1	860.457899244198\\
2	1029.93071638787\\
3	1187.30067034075\\
4	1330.33147352291\\
5	1451.62603299262\\
6	1549.48038419489\\
7	1624.21390068205\\
8	1677.32380484133\\
9	1714.42347812523\\
10	1741.41268857651\\
11	1762.31584903504\\
12	1778.34685840449\\
13	1790.61711002298\\
14	1800.67748533418\\
15	1809.63259091113\\
16	1817.70148108611\\
17	1825.02135098073\\
18	1831.69747562708\\
19	1837.81761240179\\
20	1843.45832007501\\
21	1848.68916093226\\
22	1853.57623861324\\
23	1858.18733076081\\
24	1862.60113326602\\
25	1866.92745786017\\
26	1871.36074936936\\
27	1876.39481153647\\
28	1883.80103880758\\
29	1896.18209172185\\
30	1903.66854899202\\
31	1908.16715985828\\
32	1911.74243209231\\
33	1915.1272890829\\
34	1918.37827058199\\
35	1921.51558920372\\
36	1924.55168082064\\
37	1927.49658867376\\
38	1930.35947014171\\
39	1933.14767294437\\
40	1935.86843264608\\
41	1938.52814146511\\
42	1941.13247466007\\
43	1943.68651042936\\
44	1946.19526550865\\
45	1948.66231436633\\
46	1951.09157881524\\
47	1953.48573723398\\
48	1955.84797658616\\
49	1958.1802247756\\
50	1960.48433587745\\
51	1962.76128797168\\
52	1965.01250050412\\
53	1967.23903260506\\
54	1969.44128979038\\
55	1971.61948736753\\
56	1973.77414138079\\
57	1975.90465459267\\
58	1978.01230669046\\
59	1980.09643154044\\
60	1982.15657796754\\
61	1984.19416889132\\
62	1986.20682120809\\
63	1988.19715696241\\
64	1990.16409973492\\
65	1992.10735743489\\
66	1994.0280435302\\
67	1995.92573137587\\
68	1997.80150107461\\
69	1999.65461021276\\
70	2001.4877037004\\
71	2003.29910382661\\
72	2005.08927101208\\
73	2006.86034989795\\
74	2006.1462368085\\
75	2010.34764357132\\
76	2012.06636844722\\
77	2013.76843545175\\
78	2015.4537587845\\
79	2017.12616626509\\
80	2018.78516765341\\
81	2020.43333323414\\
82	2022.07104571406\\
83	2023.69722573583\\
84	2025.31624573166\\
85	2026.93011346415\\
86	2028.5388523376\\
87	2030.1450740486\\
88	2031.74898115283\\
89	2033.35318644296\\
90	2034.96213763824\\
91	2036.5734751435\\
92	2038.18692642268\\
93	2039.81243768338\\
94	2041.44484768987\\
95	2043.09276776475\\
96	2044.74976553689\\
97	2046.42980837092\\
98	2048.11453551092\\
99	2049.82255207027\\
};
\addlegendentry{$C_n = 950$}

\end{axis}
\end{tikzpicture}%

%% file: Figures/hist14.tex
\begin{tikzpicture}

\begin{axis}[%
width=3.398in,
height=3.527in,
at={(0.758in,0.481in)},
scale = 1 , %only axis
xmin=0,
xmax=35,
xlabel style={font=\color{white!15!black}},
xlabel={Iteration},
ymin=500,
ymax=3000,
ylabel style={font=\color{white!15!black}},
ylabel={Achievable Sum-Rate in Mbps},
axis background/.style={fill=white},
xmajorgrids,
ymajorgrids,
legend style={at={(0.97,0.03)}, anchor=south east, legend cell align=left, align=left, draw=white!15!black}
]
\addplot [color=black, line width=1.2pt, mark=diamond, mark options={solid, black}]
  table[row sep=crcr]{%
1	677.262328340869\\
2	698.915734699689\\
3	699.991916323622\\
4	699.999993550253\\
5	700.000000128646\\
6	700.000000128646\\
7	700.000000128646\\
8	700.000000128646\\
9	700.000000128646\\
10	700.000000128646\\
11	700.000000128646\\
12	700.000000128646\\
13	700.000000128646\\
14	700.000000128646\\
15	700.000000128646\\
16	700.000000128646\\
17	700.000000128646\\
18	700.000000128646\\
19	700.000000128646\\
20	700.000000128646\\
21	700.000000128646\\
22	700.000000128646\\
23	700.000000128646\\
24	700.000000128646\\
25	700.000000128646\\
26	700.000000128646\\
27	700.000000128646\\
28	700.000000128646\\
29	700.000000128646\\
30	700.000000128646\\
31	700.000000128646\\
32	700.000000128646\\
33	700.000000128646\\
34	700.000000128646\\
};
\addlegendentry{$C_n = 100$}

\addplot [color=red,line width=1.2pt, mark=triangle, mark options={solid, red}]
  table[row sep=crcr]{%
1	1832.0879086311\\
2	1902.66644345583\\
3	1931.07278414195\\
4	1944.86703453013\\
5	1954.0595729908\\
6	1961.70255226801\\
7	1968.11279217006\\
8	1973.61593257303\\
9	1978.29451281459\\
10	1982.50204925628\\
11	1986.36338013052\\
12	1989.98109605599\\
13	1993.22163073889\\
14	1996.20074409335\\
15	1998.97598627718\\
16	2001.61623084836\\
17	2004.16372408433\\
18	2006.61917080208\\
19	2008.89053369071\\
20	2010.97765404021\\
21	2013.00148684068\\
22	2015.00323872998\\
23	2015.00323872998\\
24	2015.00323872998\\
25	2015.00323872998\\
26	2015.00323872998\\
27	2015.00323872998\\
28	2015.00323872998\\
29	2015.00323872998\\
30	2015.00323872998\\
31	2015.00323872998\\
32	2015.00323872998\\
33	2015.00323872998\\
34	2015.00323872998\\
};
\addlegendentry{$C_n = 450$}

\addplot [color=blue,line width=1.2pt, mark=o, mark options={solid, blue}]
  table[row sep=crcr]{%
1	2559.61123943946\\
2	2594.95051113222\\
3	2609.49337183204\\
4	2620.31776722451\\
5	2629.76314853473\\
6	2637.8175548786\\
7	2644.95992124075\\
8	2651.39072498231\\
9	2656.9471097905\\
10	2662.07101012118\\
11	2666.80789249427\\
12	2671.24636396288\\
13	2675.43451195458\\
14	2679.39417764289\\
15	2683.14645804294\\
16	2686.7550075249\\
17	2690.18319023515\\
18	2693.47254644194\\
19	2696.65389594121\\
20	2699.71370614226\\
21	2702.64496636371\\
22	2705.47575796762\\
23	2708.21636267019\\
24	2710.84940202461\\
25	2713.37366770918\\
26	2715.8242829621\\
27	2718.20613697473\\
28	2720.52043454039\\
29	2722.77607866776\\
30	2724.97956741962\\
31	2727.12525283151\\
32	2729.23030295982\\
33	2731.29988452629\\
34	2733.32106634725\\
};
\addlegendentry{$C_n = 950$}

\end{axis}
\end{tikzpicture}%

%% file: Figures/figureA12.tex
\begin{tikzpicture}

\begin{axis}[%
width=3.398in,
height=3.527in,
at={(0.758in,0.481in)},
scale = 1 , %only axis
xmin=5,
xmax=45,
xlabel style={font=\color{white!15!black}},
xlabel={Number of users},
ymin=200,
ymax=2200,
ylabel style={font=\color{white!15!black}},
ylabel={Achievable Sum-Rate in Mbps},
axis background/.style={fill=white},
xmajorgrids,
ymajorgrids,
legend style={at={(0.97,0.03)}, anchor=south east, legend cell align=left, align=left, draw=white!15!black}
]
\addplot [color=black, line width=1.2pt, mark=square, mark options={solid, black}]
  table[row sep=crcr]{%
5	389.405315776402\\
10	757.564891257741\\
15	1131.05046560022\\
20	1297.16638499173\\
25	1326.3215770004\\
30	1380.36017376183\\
35	1385.236518771\\
40	1383.42191198398\\
45	1382.11050193905\\
};
\addlegendentry{Dynamic: TIN}

\addplot [color=blue, line width=1.2pt, mark=o, mark options={solid, blue}]
  table[row sep=crcr]{%
5	706.168875384912\\
10	1319.30282874843\\
15	1752.05708213055\\
20	1991.84404661629\\
25	2075.96162333486\\
30	2076.94328309898\\
35	2085.18556755186\\
40	2109.64003687744\\
45	2128.72825638294\\
};
\addlegendentry{$\text{Dynamic: RS-CMD, }\mu\text{ = 25}$}

\end{axis}

\begin{axis}[%
width=5.933in,
height=4.692in,
at={(0in,0in)},
scale only axis,
xmin=0,
xmax=1,
ymin=0,
ymax=1,
axis line style={draw=none},
ticks=none,
axis x line*=bottom,
axis y line*=left,
legend style={legend cell align=left, align=left, draw=white!15!black}
]
\end{axis}
\end{tikzpicture}%